\documentclass[a4paper,10pt]{scrartcl}
\usepackage{amssymb}
\usepackage{array}
\usepackage{graphicx}
\usepackage[english]{babel}
\usepackage[utf8]{inputenc}
\usepackage{amsmath}
\usepackage{amsthm}
\usepackage{units}
\usepackage[arrow, matrix, curve]{xy}
\usepackage[left,modulo]{lineno}
\newcommand{\tu}[1]{\textup{#1}}

\newcommand{\Abb}[4]{\left\{ \begin{array}{ccc}
                               #1 & \rightarrow &#2\\
			       #3 &\mapsto &#4
                               \end{array}\right.}

\newcommand{\Tr}{\textup{Tr}}

\newcommand{\N}{\mathbb{N}}
\newcommand{\R}{\mathbb R}
\newcommand{\Z}{\mathbb Z}
\newcommand{\C}{\mathbb{C}}

\theoremstyle{plain}
\newtheorem{Satz}{Satz}
\newtheorem{lem}[Satz]{Lemma}
\newtheorem{prop}[Satz]{Proposition}
\newtheorem{thm}[Satz]{Theorem}
\newtheorem{cor}[Satz]{Corrolary}
\theoremstyle{definition}

\newtheorem{Def}[Satz]{Definition}
\newtheorem{exmpl}[Satz]{Example}
\theoremstyle{remark}
\newtheorem{rem}[Satz]{Remark}

\newtheorem{innercustomthm}{\bf Theorem}
\newenvironment{customthm}[1]
  {\renewcommand\theinnercustomthm{#1}\innercustomthm}
  {\endinnercustomthm}

\numberwithin{equation}{section}
\numberwithin{Satz}{section}
\numberwithin{figure}{section}

\title{Resonance chains and geometric limits on Schottky surfaces}
\author{Tobias Weich}

\begin{document}
\maketitle

\begin{abstract}
Resonance chains have been observed in many different  physical and mathematical 
scattering problems. Recently numerical studies linked the phenomenon of 
resonances chains to an approximate clustering of the length spectrum on integer 
multiples of a base length. A canonical example of such a scattering system is 
provided by 3-funneled hyperbolic surfaces where the lengths of the three 
geodesics around the funnels have rational ratios.  In this article we present 
a mathematical rigorous study of the resonance chains for these systems. 
We prove the analyticity of the generalized zeta function which provide the 
central mathematical tool for understanding the resonance chains. Furthermore 
we prove for a fixed ratio between the funnel lengths and in the limit of 
large lengths that after a suitable rescaling, the resonances in a bounded 
domain align equidistantly along certain lines. The position of these lines 
is given by the zeros of an explicit polynomial which only depends on the ratio 
of the funnel lengths.
\end{abstract}
\tableofcontents
\section{Introduction}
Let $X=\Gamma\backslash\mathbb H$ be a convex co-compact hyperbolic surface, 
then this surface has infinite volume, finite genus and a finite number of 
funnels. The resolvent of the positive Laplacian $\Delta_X$ is usually defined
as 
\[
 R(s)=(\Delta_X-s(s-1))^{-1},
\]
and on $L^2(X)$ it is analytic in $s$ for $\tu{Re}(s)>1$. Changing the function 
spaces this resolvent can be meromorphically extended to $s\in\C$ with poles of 
finite rank. The poles of this meromorphic continuation are called the resonances
of $X$ and the multiplicity of a resonance is defined by the rank of the 
associated pole. The set of all resonances on $X$ repeated according to 
multiplicity will be called $\tu{Res}(X)$.

The study of the distribution of resonances on infinite volume hyperbolic 
surfaces is of interest in number theory (see e.g. the recent work of 
Bourgain-Gamburd-Sarnak \cite{BGS11} on the affine sieve) as well as in the 
study of quantum chaos, because these surfaces provide an important model of
open, classically chaotic systems (see \cite{Non11} for a recent review). 

Since the seminal work of Patterson it is
known that there is alway one resonance at $s=\delta$, where $\delta$ is the 
Hausdorff dimension of the limit set $\Lambda(\Gamma)$(see \cite{Pat76} for 
$\delta >1/2$ and \cite{Pat88} for $\delta \leq 1/2$). For the hyperbolic 
cylinder even the complete resonance spectrum can be computed, but 
apart from this special example there are no other explicit formulas for the location 
of individual resonances. 

However, it has been a very fruitful approach to prove coarser results on the 
distribution of resonances in the complex plane. For example 
Guillop\'e-Lin-Zworski \cite{GLZ04} proved a fractal Weyl upper bound
on the number of resonances near the critical line
\[
 \#\left\{s\in\tu{Res}(X),~r\leq|\tu{Im}(s)|\leq r+1 \tu{ and Re}(s)>-C \right\}=O(r^\delta)
\]
and Naud \cite{Nau05} the existence of a spectral gap, i.e.\,of a constant 
$\varepsilon>0$ such that
\[
 \{s\in\tu{Res}(X),~\tu{Re}(s)>\delta-\epsilon\}=\{\delta\}.
\]
Such asymptotic results on the resonance distribution have important analogons in
theoretical physics \cite{GR89sc,LSZ03,ST04} and are even observable experimentally \cite{BWPSKZ13,PWBKSZ12}.
Despite the big progress in recent years there are still many open questions and 
conjectures. For example it has been conjectured that the fractal Weyl upper 
bound is sharp and that, in the semiclassical limit, i.e.\,for $\tu{Im}(s)\to\infty$
the spectral gap can be extended to $\delta/2$. For a more thorough discussion on 
recent progress on the distribution of resonances and open questions we refer 
to \cite{Non11} and references therein. 

The existence of these open conjectures motivated Borthwick to study the resonance
spectrum on infinite volume hyperbolic surfaces numerically \cite{Bor14} and to a great 
surprise he observed that the resonances on 3-funneled Schottky surfaces are
often highly ordered and form resonance chains. It has recently been shown 
in a numerical study by Barkhofen, Faure and the author \cite{phys_art} that 
these resonance chains 
can be understood by a generalized zeta function and are related to a clustering
of the length spectrum on $X$ and that the same kind of resonance chains
also appear in various other physical systems (for details on the significance of resonance chains in physics we refer to \cite{WBKPS14} and references therein). 

{\emph{Statement of the results:}} In this article we will show the existence
of resonance chains for 3-funneled surfaces by proving explicit formulas for
individual resonances in a certain geometrical limit in the Teichmüller space. 
A 3-funneled Schottky surface of genus zero is up to isometry uniquely defined
by three positive real numbers $l_1,l_2,l_3$ which correspond to the funnel widths
i.e.\,the lengths of the primitive closed geodesics that turns once around one funnel.
The numbers $l_1,l_2,l_3$ are also called Fenchel-Nielsen coordinates 
of the Teichmüller space of 3-funneled surfaces (cf. \cite[Section 13.3]{Bor07}) 
and we will denote these surfaces 
by $X_{l_1,l_2,l_3}$.  Let $n_1,n_2,n_3\in \N$
be positive integers, then we consider for $\ell>0$ the family of Schottky surfaces
\[
X_{n_1,n_2,n_3}(\ell):=X_{n_1\ell,n_2\ell,n_3\ell} 
\]
with a fixed rational ratio 
of the funnel widths. In the limit $\ell \to\infty$ the system becomes more and 
more open and the dimension of the limit set tends to zero $\delta\to0$. One observes 
that in this limit not only the leading resonance at $\delta$ tends to zero but 
all other resonances do as well. In order to study a meaningful, 
nontrivial limit of the resonance spectrum, the spectrum has to be rescaled with $\ell$ 
and we define the set of rescaled resonances as
\[
 \widetilde{\tu{Res}}_{n_1,n_2,n_3}(\ell):=\{s\in\C,~s/\ell \in 
 \tu{Res}(X_{n_1,n_2,n_3}(\ell))\}.
\]
Then we obtain the following theorem for the rescaled resonances in the limit 
$\ell\to \infty$.
\begin{thm}\label{thm:location_res}
 Let $n_1,n_2,n_3$ be positive integers that fulfill a triangle condition
 i.e.\,they fulfill the inequality 
 $n_i+n_j>n_k$ for any permutation of $1,2,3$. Let furthermore be
\begin{equation}\label{eq:polynomial}
  \begin{array}{rcl}P_{n_1,n_2,n_3}(x)&:=& 1-2(x^{n_1}+x^{n_2}+x^{n_3})+
  x^{2n_1}+x^{2n_2}+x^{2n_3}\\
   &&+2(x^{n_1+n_2}+x^{n_2+n_3}+x^{n_1+n_3})-4x^{n_1+n_2+n_3}
  \end{array}
 \end{equation}
and 
\[
\mathcal N_{n_1,n_2,n_3}:=\{s\in\C,~P_{n_1,n_2,n_3}(e^{-s})=0\}
\]
where the zeros are repeated according to the multiplicities.
Then for any bounded domain $U\subset \C$ with 
$\partial U\cap \mathcal N_{n_1,n_2,n_3}=\emptyset$ we have
\[
 \lim\limits_{\ell\to\infty} \#\Big( U\cap\widetilde{\tu{Res}}_{n_1,n_2,n_3}(\ell)\Big)
 = \#\Big( U\cap\mathcal N_{n_1,n_2,n_3}\Big).
\]
\end{thm}
Note that $U$ can be chosen arbitrarily small, so Theorem~\ref{thm:location_res}
states that a finite number of resonances is determined by $P_{n_1,n_2,n_3}$ at 
an arbitrary precision for large enough $\ell$. As $\mathcal N_{n_1,n_2,n_3}$ is 
the zero set of a polynomial in $e^{-s}$, this set naturally forms  
straight chains in the sense that 
\[
s_0\in\mathcal N_{n_1,n_2,n_3}\Rightarrow s_k=s_0+2\pi ik \in\mathcal N_{n_1,n_2,n_3},~~ 
\forall~ k\in\Z.
\]
Theorem~\ref{thm:location_res} then says that the rescaled resonance spectrum 
converges against straight resonance chains which are explicitly described by the 
polynomial $P_{n_1,n_2,n_3}$. Note that this convergence however only holds 
for an arbitrarily large but finite number of resonances and one can 
not suspect Theorem~\ref{thm:location_res} to hold uniformly for all resonances
because this would contradict the fractal Weyl conjecture on the number of resonances
in the semiclassical limit. The limit $\ell\to\infty$ thus can be understood as a 
limit complementary to the semiclassical limit which holds in the low frequency regime
i.e.\,for a finite number of resonances. And in fact we will see in Section~\ref{sec:num}
that $P_{n_1,n_2,n_3}$ describes the first 50-100 resonances 
already for relatively small values of $\ell\approx 4$.

In the proof of Theorem~\ref{thm:location_res} a generalized dynamical zeta function
will play an important role and we will show that such generalized zeta functions
always have an analytic extension. Therefore we introduce $\mathcal P_{X_{l_1,l_2,l_3}}$
as the set of all primitive closed geodesics on $X_{l_1,l_2,l_3}$, where primitive means
that the geodesic is not a repetition of a shorter closed geodesic. If we additionally 
denote for a closed geodesic $\gamma$ its length by $l(\gamma)$ then we can state the 
following result.
\begin{thm}\label{thm:gen_zeta_analytic}
 Let $X_{l_1,l_2,l_3}$ be a Schottky surface with three funnels of widths 
 $l_1,l_2,l_3$ and let $n_1,n_2,n_3\in \N$. We define 
 \begin{equation}\label{eq:order_function_on_geodesics}
\mathbf{n}:\Abb{\mathcal P_{X_{l_1,l_2,l_3}}}{\N}{\gamma}{\sum_{i=1}^3 n_iw_i(\gamma)}  
 \end{equation}
where $w_i(\gamma)$ denotes the winding number around the funnel of width $l_i$. 
Then the generalized zeta function
\begin{equation}\label{eq:gen_zeta:limit}
 d_\mathbf{n}(s,z) = \prod\limits_{\gamma\in \mathcal P_{X_{l_1,l_2,l_3}}}
                       \prod\limits_{k\geq0}                       
                       \left(1-z^{\mathbf{n}(\gamma)}e^{-(k+s)l(\gamma)}\right).
\end{equation}
extends to an analytic function on $\C^2$.
\end{thm}
Similar to an ordinary dynamical zeta function, obtained by a Bowen-Series transfer operator,
this generalized zeta function is equal to the Selberg zeta function of $X_{l_1,l_2,l_3}$
for $z=1$. Beside its appearance in the proof of Theorem~\ref{thm:location_res} this
result is also of independent interest as in \cite{phys_art} it has been numerically 
shown that for the understanding of the resonance chains for finite $\ell$  these 
generalized zeta functions are the central object. The numerical algorithms used in
\cite{phys_art} and the interpretation of the results were also heavily based on the 
assumption that the generalized zeta function is analytic.

The particularly simple structure of the resonance spectrum in the limit
$\ell\to\infty$ as stated in Theorem~\ref{thm:location_res} can finally be 
understood by the following result which states the in the limit $\ell\to \infty$
the generalized zeta function of Theorem~\ref{thm:gen_zeta_analytic}
is given by the polynomial $P_{n_1,n_2,n_3}$.
\begin{thm}\label{thm:rescaled_zeta_limit}
 Let $n_1,n_2,n_3$ be positive integers fulfilling the triangle condition.  
 Consider for $\ell>0$ the family of Schottky surfaces 
 $X_{n_1,n_2,n_3}(\ell)$ then the 
 generalized zeta function of this family of surfaces, as defined 
 in Theorem~\ref{thm:gen_zeta_analytic}, also depends on the 
 parameter $\ell$ and we denote it by $d_{\mathbf{n}}(s,z;\ell)$. If 
 $P_{n_1,n_2,n_3}$ is the polynomial defined in (\ref{eq:polynomial}),
 then on any bounded set $B\subset \C^2$ the rescaled generalized 
zeta function $d_{\mathbf{n}}(z,s/\ell;\ell)$ converges to 
the polynomial in the sense that
\begin{equation}\label{eq:rescaled_zeta_polynomial_convergence}
 \lim\limits_{\ell\to\infty}\left\| d_{\mathbf{n}}(z,s/\ell;\ell) - P_{n_1,n_2,n_3}(ze^{-s})\right\|_{\infty,B}.
\end{equation}
\end{thm}
The article is organized as follows: First we will recall some basic facts on the
definition of Schottky surfaces, their resonances and Selberg zeta functions in 
Section~\ref{sec:res_Schottky}. Then, in Section~\ref{sec:dyn_zeta}, we recall
the definition of the dynamical zeta function the way they are 
usually obtained using Bowen-Series
maps and iterated function schemes (IFS). While this traditional 
approach is very natural from an algebraic point of view, we will see that it is not 
natural from the geodesic flow point of view. Section~\ref{sec:gen_zeta:limit} will then 
be dedicated to an iterated function scheme whose dynamical zeta function 
also contain the Selberg zeta function of the Schottky surface but that is 
much better adapted to the geodesic flow. These \emph{flow-adapted IFS} are then
used to prove Theorem~\ref{thm:gen_zeta_analytic} on the analyticity of the 
generalized zeta functions. The flow-adapted IFS in addition turn out to 
be the central ingredient in treating the limit $\ell\to\infty$ in 
Section~\ref{sec:geom_limit}. The idea in proving Theorem~\ref{thm:location_res} 
and Theorem~\ref{thm:rescaled_zeta_limit} is to write the generalized zeta
function as a Fredholm determinant of a transfer operator defined by the 
flow-adapted IFS. If one then considers the Taylor expansion of this Fredholm
determinant it can be shown using techniques from Jenkinson-Pollicott \cite{JP02}, 
that in the limit $\ell\to\infty$ only the first view terms survive. Furthermore
all remaining terms become particularly simple and cancel each other to a great 
extend. From a physical point of view the proof strategy is to show that the ideas 
which Cvitanovic-Eckhardt \cite{CE89} introduced under the name ``cycle expansion'' 
in physics
become rigorous in the limit $\ell\to\infty$ on Schottky surfaces. 
Finally, in Section~\ref{sec:num} we will compare the results
with numerical calculations and we will observe that the resonances in the 
low-frequency regime are already surprisingly well described by 
$P_{n_1,n_2,n_3}$ for relatively small values of $\ell$ ($\ell\approx 4$) which illustrates
the practical value of Theorem~\ref{thm:location_res}.

\emph{Acknowledgements:} I am grateful to Frédéric Faure who proposed to study these geometric limits and motivated me to start this work. The discussions with
him were a constant source of inspiration throughout this work. 
I am also thankful to David Borthwick and
Pablo Ramacher for helpful discussions and corrections of an early stage of 
this article. This work was supported by the German National Academic Foundation 
and by the Agence National de Recherche via the project 2009-12 METHCHAOS. 

\section{Resonances and zeta functions for Schottky surfaces}
\label{sec:res_Schottky}
All hyperbolic surfaces can be written as a quotient of the hyperbolic half 
plane $\mathbb H$ by a discrete subgroup of its orientation preserving 
isometry group 
$\Gamma\subset \tu{Isom}^+(\mathbb H)=PSL(2,\R)$. We will be particularly 
interested in Schottky surfaces which are quotients by certain freely 
generated groups, called Schottky groups. These groups can be defined as follows.
\begin{Def}\label{def:Schottky_group}
Let $D_1,\dots,D_{2r}$ be disjoint open disks in $\C$ with
centers on the real line and mutually disjoint closures. Then there exists for each
pair $D_i, D_{i+r}$ a hyperbolic element $S_i\in PSL(2,\R)$ that maps $\partial D_i$ to 
$\partial D_{i+r}$ and that maps the interior of $D_i$ to the exterior of 
$D_{i+r}$. A \emph{Schottky group} is then the free group generated by $S_1,\dots,S_r$. 
\end{Def}
With this definition Schottky surfaces are always surfaces of infinite 
volume without cusps and with a finite number of funnels. The simplest 
nontrivial example of Schottky surfaces are those surfaces with three 
funnels of genus zero (see upper part or Figure~\ref{fig:poinc_bowen_series}). 
Given three positive real numbers $l_1,l_2,l_3$ a Schottky group of such a
surface is freely generated by the two hyperbolic elements
\begin{equation}\label{eq:s1_s2_def}
 S_1=\left(\begin{array}{cc}
            \cosh(l_1/2)&\sinh(l_1/2)\\
            \sinh(l_1/2)&\cosh(l_1/2)
           \end{array}
\right),~~~ S_2=\left(\begin{array}{cc}
           \cosh(l_2/2)&a\sinh(l_2/2)\\
            a^{-1}\sinh(l_2/2)&\cosh(l_2/2)
           \end{array}
\right), 
\end{equation}
where the parameter $a$ is chosen such that $\Tr(S_1S_2^{-1})=-2\cosh(l_3/2)$.
We write
\[
\Gamma_{l_1,l_2,l_3}:=\langle S_1,S_2\rangle \tu{ and } X_{l_1,l_2,l_3}:=\Gamma_{l_1,l_2,l_3} \backslash \mathbb H. 
\] 
The parameters $l_1,l_2,l_3$ coincide with the lengths of the three primitive 
closed geodesics around the three funnels of the surface $X_{l_1,l_2,l_3}$ 
(see the geodesics $\gamma_1,\gamma_2$ and $\gamma_3$ in Figure~\ref{fig:poinc_bowen_series})
and parametrize uniquely all hyperbolic surfaces of this type. They are also 
called Fenchel-Nielsen coordinates because they are global coordinates on 
the Teichmüller space for 3-funneled surfaces of genus zero, i.e.\,the space 
of all isometry classes of hyperbolic metrics on this surface 
(cf.\,\cite[Section 13.3]{Bor07}). 

The spectral properties of a general Schottky surface $X$ are described by the 
positive Laplacian $\Delta_X$. As the surface has infinite volume it is known 
that it has at most finitely many $L^2$-eigenvalues in $(0,1/4)$ and absolutely 
continuous spectrum on $[1/4,\infty)$ with no embedded eigenvalues. The 
$L^2$-spectrum thus is not a good spectral quantity and it is well known that instead 
one has to study the resonances of the Laplace operator. These resonances
can be defined by 
the meromorphic continuation of the resolvent which has been shown by 
Mazzeo-Melrose \cite{MM87} and Guillope-Zworski \cite{GZ95}
\begin{thm}
 The resolvent
 \[
  R_X(s):=(\Delta_X-s(1-s))^{-1}: L^2(X)\to H^2(X)
 \]
which is defined for $\tu{Re}(s)\geq 1/2$ and $s(1-s)\notin \tu{spec}(\Delta_X)$ 
extends to a meromorphic family of operators
\[
 R_X(s):L^2_{\mathrm{comp}}(X) \to H^2_{\mathrm{loc}}(X)
\]
with poles of finite rank.
\end{thm}
By this meromorphic continuation we can define the set of resonances as
\begin{equation}
 \tu{Res}(X):=\{s\in \C,~s\tu{ is pole of }R_X(s), \tu{ repeated according to multiplicity}\}.
\end{equation}
Surfaces with constant negative curvature have the remarkable property that 
their resonance spectrum is related to the zeros of their Selberg zeta function 
which we introduce now. If $\gamma$ is a 
closed geodesic on a hyperbolic surface surface $X$ we can create longer closed 
geodesics by simply repeating it. We call a geodesic \emph{primitive} if it 
cannot be obtained as a repetition of a shorter geodesic and we denote the 
\emph{set of primitive closed geodesics} on $X$ by
\[
 \mathcal P_X:=\{\gamma,\tu{ closed primitive godesic on }X \}.
\]
If $l(\gamma)$ denotes the length of a geodesic $\gamma$, then the 
\emph{Selberg zeta} function of $X$ is defined as
\begin{equation}\label{eq:Selberg_zeta} 
 Z_X(s):=\prod\limits_{\gamma\in \mathcal P_X}\prod\limits_{k\geq0} \left(1-e^{-(s+k)l(\gamma)}\right).
\end{equation}
This product is absolutely convergent for $\tu{Re}(s)$ sufficiently large and for 
Schottky surfaces it extends to an analytic function on $\C$ \cite{Gui92}. The 
result of Patterson-Perry \cite{PP01} which was later generalized to surfaces 
with cusps by Borthwick-Judge-Perry \cite{BJP05} (see also Bunke-Olbrich \cite{BO99})
relates the resonances to the 
zeros of the Selberg zeta function.
\begin{thm}\label{thm:Patterson-Perry}
 For a Schottky surface $X=\Gamma\backslash \mathbb H$ the zero set of the 
 zeta function $Z_X(s)$ is the union of the resonances $\tu{Res}(X)$ and the 
 negative integers $s=-k$, $k\in \mathbb N_0$. 
\end{thm}
\section{Dynamical zeta functions for iterated function schemes}\label{sec:dyn_zeta}
The correspondence between the zeros of the Selberg zeta function and the 
resonances as stated in Theorem~\ref{thm:Patterson-Perry} is a central ingredient 
for understanding the resonance chains. However we first have to develop a 
different point of view on the Selberg zeta function by the so called dynamical 
zeta function, which we introduce in this section for holomorphic 
iterated function schemes. 
\begin{Def}[Holomorphic iterated function scheme]
 For $N\in \N$ let $D_1,\dots,D_N\subset \C$ be $N$ open disks such that their closures
 $\overline D_i$ are pairwise disjoint. Let $A\in \{0,1\}^{N\times N}$ be the
 \emph{adjacency matrix} and denote $i\rightsquigarrow j$ if $A_{i,j}=1$. Furthermore
 for each pair $(i,j)\in \{1,\dots,N\}^2$ with $i\rightsquigarrow j$ we have a
 biholomorphic map $\phi_{i,j}:D_i\mapsto \phi_{i,j}(D_i)$ such that $\phi_{i,j}
 (D_i)\Subset D_j$ and such that different images are pairwise disjoint, i.e.
 \begin{equation}\label{eq:separtaionCondition:limit}
  \phi_{i,j}(D_i)\cap\phi_{k,l}(D_k)\neq\emptyset \Leftrightarrow i=k \tu{ and }j=l.
 \end{equation}
 Finally we call a holomorphic IFS \emph{eventually contracting}, if there is a $N_0$ 
 and $\theta<1$ such that for $n\geq N_0$
 \[
  |\phi_w'(u)|\leq\theta \tu{ for all }w\in\mathcal W_n \tu{ and } u\in D_{w_0}.
 \]
\end{Def}
For convenience we denote the union of all the disjoint disks by
\[
 D:=\bigcup_i D_i
\]
and the union of all their images by
\[
 \phi(D):=\bigcup\limits_{i\rightsquigarrow j} \phi_{i,j}(D_i).
\]
From (\ref{eq:separtaionCondition:limit}) it follows directly that for $u\in\phi(D)$
there is exactly one pair $i\rightsquigarrow j$ and $u'\in D_i$ such that $u=\phi_{i,j}
(u')$. We have thus a well defined holomorphic inverse function
\[
 \phi^{-1}:\phi(D)\to D.
\]

\begin{exmpl}\label{exmpl:genSchottky:limit}

The disks $D_i$ and generators $S_i$ in the construction of a Schottky group 
(see Definition~\ref{def:Schottky_group}) give a natural
construction of a holomorphic IFS. For convenience we denote for $i=1,\dots,r$ 
$S_ {i+r} := S_i^{-1}$ and use a cyclic notation of the indices i.e.\,$S_{i+2r} = S_i$ and
$D_{i+2r} = D_i$. Then for $i=1,\dots,r$ all elements $S_i$ map all disks, except $D_i$,
holomorphically into the interior of $D_{i+r}$. Thus the adjacency matrix of this IFS
is given by a $2r\times 2r$ matrix with $A_{i,j}=0$ if $|i-j|= r$ and $A_{i,j}=1$ else. 
For any $i\rightsquigarrow j$ the maps are given by
\[
 \phi_{i,j}(u):=S_{j+r}u=S_j^{-1} u.
\]
They clearly fulfill (\ref{eq:separtaionCondition:limit}) and are also known 
to be eventually contracting (see e.g. \cite[Proposition 15.4]{Bor07}). 
Note that the inverse map $\phi^{-1}$ restricted to $D_j\cap \phi(D)$ is exactly given by $S_j$.
The IFS which we defined is consequently the inverse of the usual \emph{Bowen-Series
map} for Schottky groups (see e.g. \cite[Section 15.2]{Bor07}). 
\end{exmpl}
It will turn out to be useful for the notation to introduce the following symbolic
coding. The  \emph{symbols} are given by the integers $1,\dots,N$ and the set of words
of length $n$ by the tuples of symbols
\[
 \mathcal W_n:=\{(w_0,\dots,w_n), w_i\rightsquigarrow w_{i+1} \tu{ for all }i=0,\dots,n-1\}.
\]
Note that our notation of \emph{word length} does not refer to the number of symbols,
but to the number of transitions $w_i\rightsquigarrow w_{i+1}$ which they indicate. For $w\in \mathcal W_n$ and
$0<k\leq n$ we define the \emph{truncated word} by
\[
 w_{0,k}:=(w_0,\dots,w_k)\in\mathcal W_k.
\]
Finally we define the iteration of the maps $\phi_{i,j}$ along a word $w\in \mathcal W_n$ as
\[
 \phi_w:=\phi_{w_{n-1},w_n}\circ\ldots\circ\phi_{w_0,w_1}:D_{w_0}\mapsto D_{w_n}
\]
and their images as
\[
 D_w:=\phi_w(D_{w_0}).
\]
Note that $D_w\Subset D_{w_n}$ and that from the separation condition
(\ref{eq:separtaionCondition:limit}) one obtains inductively for $w,w'\in \mathcal W_n$
\[
D_w\cap D_{w'}\neq\emptyset  \Leftrightarrow w=w'.
\]
We call a word $w\in \mathcal W_n$ of length $n$ \emph{closed} if $w_0=w_n$ and we
denote the set of all closed words of length $n$ by $\mathcal W_n^{cl}$.
Given a closed word $w\in \mathcal W_n^{cl}$, the map 
\[
\phi_w:D_{w_0}\to D_w\Subset  D_{w_0}
\]
of an eventually contracting IFS has a unique fixed point (see e.g. \cite[Lemma 2.3]{sym_art}) 
which we denote by $u_w$. 
If a closed word $w\in \mathcal W_n^{cl}$ of length $n$ is concatenated $k$-times with 
itself, we obtain a closed word of length $k\cdot n$
\[
 w^k=(w_0,w_1,\ldots,w_{n-1},w_0,\ldots,\ldots w_{n-1},w_0,\ldots,w_{n-1},w_0)\in\mathcal W^{cl}_{nk}.
\]
In analogy to the primitive closed geodesics, we call a word $w$ \emph{prime} if it can not 
be obtained by the repetition of a shorter word and we write
\[
 \mathcal W_n^{\tu{prime}}:=\{w\in\mathcal W_n^{cl},~w \tu{ is prime}\}.
\]
Note that as well the set of closed words as the set of prime words is invariant under the left-shift
\[
 \sigma_L:\Abb{\mathcal W_n^{cl/\tu{prime}}}{\mathcal W_n^{cl/\tu{prime}}}{(w_0,\ldots,w_{n-1},w_0)}{(w_1,\ldots,w_{n-1},w_0,w_1)}.
\]
and the right-shift
\[
 \sigma_R:\Abb{\mathcal W_n^{cl/\tu{prime}}}{\mathcal W_n^{cl/\tu{prime}}}{(w_0,\ldots,w_{n-1},w_0)}{(w_{n-1},w_0\ldots,w_{n-1},w_{n-1})}.
\]
obviously $\sigma_L^{-1}=\sigma_R$ and iterative application of these 
operators induce a $\Z$-action on the set of words. The importance
of this shift action
arises from the fact that on the periodic orbits, the dynamics of the IFS is 
conjugated to the dynamics of the shift operator on the closed words in the following
sense
\begin{equation}\label{eq:conjugate_to_shift}
 \forall w\in \mathcal W_n^{cl}, 0<k<n:~~\phi_{w_0,k}(u_w)=u_{\sigma_L^k w}.
\end{equation}
We will denote by $[w]$ the orbit of a word $w$ by this $\Z$-action on 
$\mathcal W_n^{cl/\tu{prime}}$ and write the space of these orbits, i.e.\,
the quotient by the group action as
\[
 \left[\mathcal W_n^{cl/\tu{prime}}\right]:=\Z\backslash \mathcal W_n^{cl/\tu{prime}}.
\]

Next we define the transfer operators associated to the iterated function schemes.
\begin{Def}
Let $\mathcal A_\infty(D)$ be the Banach space of holomorphic functions on $D$ that 
are bounded on $\overline{D}$ with the supremum norm. If we have a function 
$V\in A_\infty (\phi(D))$
then we define the \emph{transfer operator} 
$\mathcal L_V:\mathcal A_\infty(D)\to\mathcal A_\infty(D)$ associated to the 
IFS by
\begin{equation}\label{eq:transfer_op_sum_decomposition}
  (\mathcal L_Vh)(u):= \sum\limits_{i=1}^N (\mathcal L^{(i)}_Vh)(u) 
 \end{equation}
where
\begin{equation}
 \mathcal L^{(i)}_V:\Abb{\mathcal A_\infty(D)}{\mathcal A_\infty(D_i)}{h(u)}{(\mathcal L_V^{(i)}h)(u):=\sum\limits_{j \tu{s.t.} i\rightsquigarrow j} V(\phi_{i,j}(u)) h(\phi_{i,j}(u))}.
\end{equation}
The sum in (\ref{eq:transfer_op_sum_decomposition}) is then understood in the 
sense that $\mathcal A_\infty(D)=\bigoplus_{i=1}^N\mathcal A_\infty(D_i)$.
\end{Def}

Given such a potential $V$, a word $w\in \mathcal W_n$ and a point $u\in D_{w_0}$, we
can define the iterated product
\begin{equation}\label{eq:iterated_product}
 V_w(u):=\prod\limits_{k=1}^n V(\phi_{w_{0,k}}(u)).
\end{equation}
A straight forward calculation of powers of the transfer operator $\mathcal L_V$ leads
to 
\[
 \left(\mathcal L_V^n h\right)(u)=\sum\limits_{w\in \mathcal W_n, \tu{s.t.} u\in D_{w_0}} V_w(u)h(\phi_w(u)),
\]
thus these iterated products naturally occur in powers of $\mathcal L_V$.

\begin{Def}\label{def:nuclear}
  An operator $\mathcal L:\mathcal B\to\mathcal B$ on a Banach space $\mathcal B$ 
  is called \emph{nuclear}, if there exist $v_n\in\mathcal B$, 
  $\alpha_n\in \mathcal B^*$ with $\|v_n\|=\|\alpha_n\|=1$ and $\lambda_n\in \C$ 
  with $\sum\limits_{n=0}^\infty |\lambda_n|<\infty$ such that 
  \begin{equation}\label{eq:nuclear_repres}
   \mathcal Lh=\sum\limits_{n=0}^\infty \lambda_n\alpha_n(h)v_n
  \end{equation}
  for any $h\in \mathcal B$. The representation (\ref{eq:nuclear_repres}) is 
  then called \emph{nuclear representation} of $\mathcal L$.
\end{Def}
It is a well known fact that these transfer operators of holomorphic IFS are nuclear operators 
(see \cite{Rue76} or \cite[Proposition 2]{JP02}, respectively) and that for eventually 
contracting IFS the trace 
can be expressed in terms of the points $u_w$.  Accordingly one can define the
\emph{dynamical zeta function} by the Fredholm determinant
\begin{equation}\label{eq:zeta_def:limit}
 d_V(z):=\det(1-z\mathcal L_V)
\end{equation}
which is an entire function on $\C$ and which can be written for $|z|$ sufficiently
small as (see e.g. \cite[(3.26)]{JP02})
\begin{equation}\label{eq:zeta_fixpoint_formula:limit}
 d_V(z)=\exp\left(-\sum\limits_{n>0} \frac{z^n}{n}\sum\limits_{w\in \mathcal W_n^{cl}} V_w(u_w)\frac{1}{1-\phi_w'(u_w)}\right).
\end{equation}
One has the following important connection between the dynamical and Selberg 
zeta function:
\begin{thm}\label{thm:dynamical_Selberg_equiv}
 Let $X=\Gamma\backslash\mathbb H$ be a Schottky surface and take the iterated function scheme
 associated to the Bowen-Series maps as defined in Example \ref{exmpl:genSchottky:limit}. 
 For $s\in \C$ define the potential $V_s(z)=[(\phi^{-1})'(z)]^{-s}$ and consider the 
 holomorphic family of nuclear operators $\mathcal L_{V_s}$. Then the dynamical 
 zeta function
 \[
  d(s,z):=\det(1-z\mathcal L_{V_s})
 \]
is holomorphic in both variables and 
\[
 Z_X(s)=d(s,1).
\]
\end{thm}
\begin{proof}
 We will only give a sketch of the proof here, considering those steps which will
 be of further importance in this article. For a detailed proof see 
 e.g.~\cite[Theorem 15.8]{Bor07}.

The proof heavily relies on a product form of the general dynamical zeta function 
which we will derive now. Expanding the last quotient in 
(\ref{eq:zeta_fixpoint_formula:limit}) as a geometric series one obtains
\[
 d_V(z)=\exp\left(-\sum_{k\geq 0} \sum\limits_{n>0}\sum\limits_{w\in \mathcal W_n^{cl}} \frac{z^n}{n} V_w(u_w)\left(\phi_w'(u_w)\right)^k\right).
\]
Next one checks that for a given class of words 
$[w]\in \left[ \mathcal W_n^{cl}\right]$ neither $V_w(u_w)$ nor 
$\phi'_w(u_w)$ depend on the choice of the representative $w$. Furthermore one 
easily calculates that
 \[
  V_{w^k}(u_{w^k})=(V_w(u_w))^k \tu{ and }\phi'_{w^k}(u_{w^k})=(\phi'_w(u_w))^k.
 \]
Consequently the sum over all closed words represented by the double 
sum $\sum\limits_{n>0}\sum\limits_{w\in \mathcal W_n^{cl}}$ can be 
transformed into a sum over all classes of prime words and their 
repetitions and one obtains
\[
 d_V(z)=\exp\left(-\sum_{k\geq 0} \sum\limits_{r> 0}\sum \limits_{[w]\in  \left[\mathcal W^{\tu{prime}}\right]} \frac{\left(z^{n_w}V_w(u_w)\left(\phi_w'(u_w)\right)^k\right)^r}{r} \right).
\]
where $\left[\mathcal W^{\tu{prime}}\right] =
\bigcup_n\left[\mathcal W^{\tu{prime}}_n\right]$ is the set of the 
prime word-classes of arbitrary length and $n_w$ denotes the length of the word $w$.
Finally using the Taylor expansion $\log(1-x)=-\sum_{r> 0} x^r/r$ 
one obtains
\begin{equation}\label{eq:zeta_product_form}
 d_V(z)=\prod \limits_{[w]\in \left[\mathcal W^{\tu{prime}}\right]}\prod_{k\geq 0}  \left(1-z^{n_w}V_w(u_w)\left(\phi_w'(u_w)\right)^k\right).
\end{equation}

With the special choice of the potential $V_s(u)=((\phi^{-1})'(u))^{-s}$ 
one then obtains
\[
 d(s,z)= \prod \limits_{[w]\in \left[\mathcal W^{\tu{prime}}\right]}\prod_{k\geq 0} \left(1-z^{n_w}\left(\phi_w'(u_w)\right)^{s+k}\right).
\]
The equivalence to the Selberg zeta function then follows from a 
one-to-one correspondence between the classes of prime words of 
the Bowen-Series IFS and the primitive geodesics on the Schottky 
surfaces (see e.g.~\cite[Proposition 15.5]{Bor07}) and from the 
fact that the stabilities of the fixed points $\phi_w'(u_w)$ are 
related to the lengths of these geodesics.
\end{proof}

As explained in Section~\ref{sec:res_Schottky} we are especially 
interested in zeros of the zeta function. The fixed point formula 
(\ref{eq:zeta_fixpoint_formula:limit}) for the dynamical zeta function 
can however not vanish if the series are absolutely convergent. 
But the fixed point formula (\ref{eq:zeta_fixpoint_formula:limit}) is
only valid in the region of absolute convergence, in the rest of the
complex plane the zeta function is only defined by analytic continuation.
Equation (\ref{eq:zeta_fixpoint_formula:limit}) is thus only valid in a region where the zeta function has no 
zeros. The same is true for the product formula (\ref{eq:Selberg_zeta}) 
of the Selberg zeta functions. Both formulas are thus not at all 
useful for determining the zeros numerically. One can however 
use the following clever trick which has been introduced in 
physics by Cvitanovic-Eckhardt \cite{CE89} under the name 
\emph{cycle expansion} and that has been rigorously 
applied to Schottky surfaces by Jenkinson-Pollicott \cite{JP02} in 
mathematics: As for any bounded potential $V$ the series in  
(\ref{eq:zeta_fixpoint_formula:limit}) converges in a neighborhood of 
zero, one can derive a general formula for the Taylor coefficients 
of the Taylor expansion of $d_V(z)$ in $z$ around zero. 
This expansion is given by \cite[Proposition 8]{JP02}:
\begin{equation}\label{eq:cycle_expansion:limit}
 d_V(z)=1+\sum\limits_{N=1}^\infty z^N d_V^{(N)}
\end{equation}
with 
\begin{equation}\label{eq:cycle_coeff_formula}
 d_V^{(N)} = \sum\limits_{m=1}^N\left(\sum\limits_{(n_1,\dots,n_m)\in P(N,m)}\frac{(-1)^m}{m!}\prod\limits_{l=1}^{m} \frac{1}{n_l} \sum\limits_{w\in \mathcal W^{cl}_{n_l}} \frac{V_w(u_w)}{1-\phi_w'(u_w)}\right)
\end{equation}
where $P(N,m)$ is the set of all $m$-partitions of $N$, i.e.\,the set of all 
integer $m$-tuples that sum up to $N$.
As $d_V(z)$ is known to be analytic on all $\C$ its Taylor expansion 
(\ref{eq:cycle_expansion:limit}) converges absolutely on $\C$ and is well 
suited for numerical calculations of its zeros (c.f. \cite{Bor14}).

\section{Flow-adapted iterated function schemes and generalized zeta functions}
\label{sec:gen_zeta:limit}
As mentioned in the proof of Theorem~\ref{thm:dynamical_Selberg_equiv} the key 
ingredient for the equivalence between the dynamical zeta function of 
the standard 
Bowen-Series IFS and the Selberg zeta function is an equivalence between 
periodic geodesics on the surface and periodic orbits of the IFS. This equivalence
is usually proven in a purely algebraic way by arguing with conjugacy classes
in the Schottky group $\Gamma$. This equivalence can however also be understood
from a geometric or dynamical point of view by interpreting the Bowen-Series maps
as some kind of Poincaré section of the geodesic flow. The 3-funneled 
Schottky surface can be obtained from its fundamental domain by gluing together 
the circles of the same color (see Figure~\ref{fig:poinc_bowen_series}).
\begin{figure}
\centering
        \includegraphics[width=0.7\textwidth]{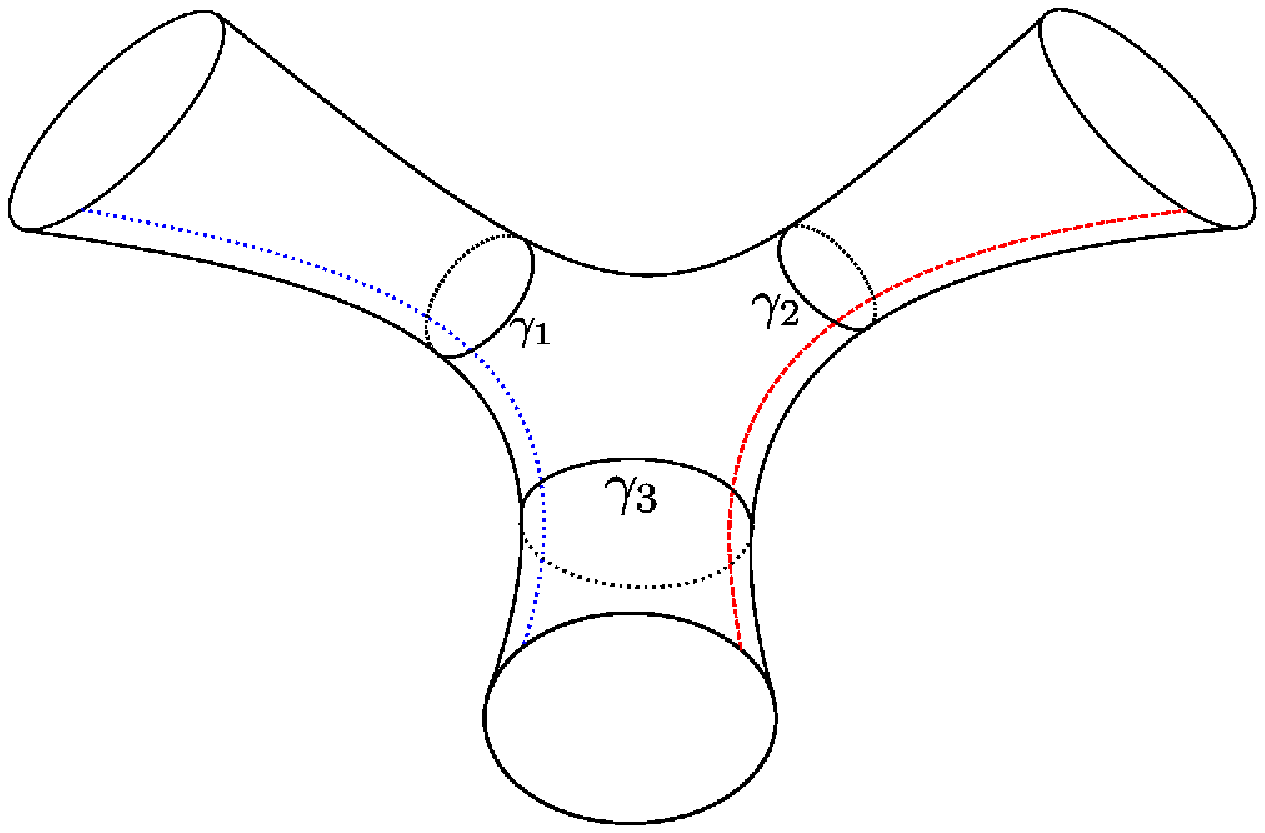}
        \includegraphics[width=\textwidth]{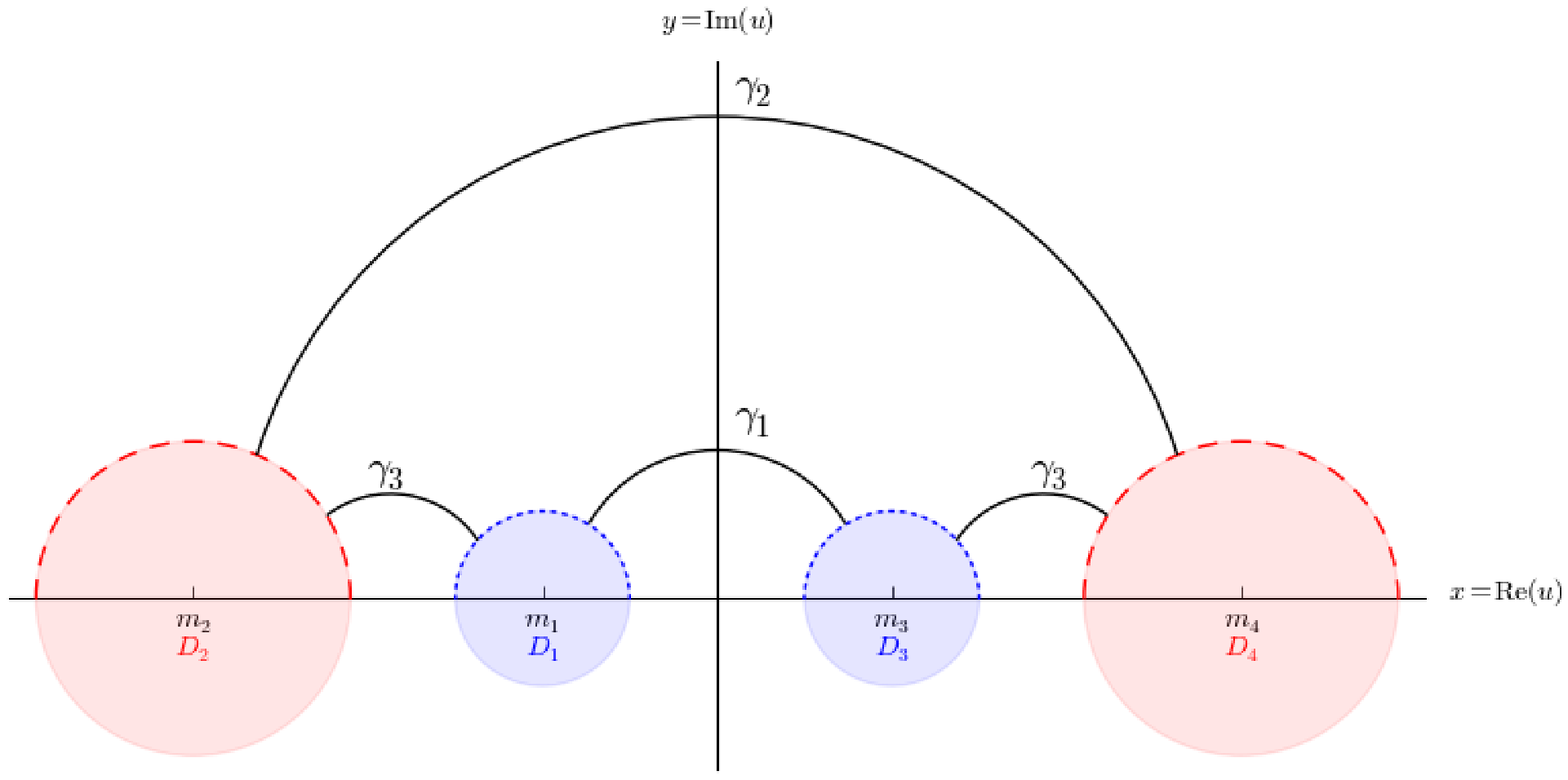}
\caption{Upper part: Schematic sketch of a 3-funneled Schottky 
surface. The dashed red and dotted blue line indicate the cut lines which 
would correspond
to the Poincaré section of the standard Bowen-Series IFS. The black
lines represent the three fundamental geodesics that wind one time around 
one of the funnels. Lower part: Configuration of 4 disks that
give rise to the construction of a 3-funneled Schottky surface. The upper half plane
without the disks represents a fundamental domain and the surface can be obtained by gluing
together the two red dashed lines and the two blue dotted lines. In black the three fundamental
closed geodesics $\gamma_1,\gamma_2,\gamma_3$ from the upper part of the figure are shown. 
While $\gamma_1$ and $\gamma_2$ are only represented by one arc each, the geodesic $\gamma_3$
appears as two arcs in the fundamental domain.}
\label{fig:poinc_bowen_series}
\end{figure}
Each closed
geodesic on the surface crosses the blue and red cut lines a finite 
number of times and can be represented by one or several arcs in the 
fundamental domain. 
The fixed points
of the Bowen-Series map are then exactly the end-points of these arcs. We do 
not want to go any further into details, as we will need no rigorous
statement
of this correspondence in the sequel (in all proofs it is more convenient to do
the calculations from the algebraic point of view). It is however important to 
realize that the standard Bowen-Series IFS seems to be very natural from the 
algebraic point of view (it is directly constructed from the two generators
$S_1, S_2$ of the freely generated Schottky group) but not from the point of 
view of the geodesic flow: Geodesics that turn one time around one of the funnels 
and which are topologically similar are treated differently depending on the choice
of the funnel. For example the geodesics $\gamma_1$ and $\gamma_2$ in 
Figure~\ref{fig:poinc_bowen_series} only cross one cut line, while the 
geodesic $\gamma_3$ crosses two of them. This implies that $\gamma_3$
corresponds to a periodic orbit of word length two while $\gamma_1$ and $\gamma_2$ only 
correspond to a word length one. From a purely dynamical point of view it would 
thus be more natural to take a Poincaré section with three cut lines as presented
on the lower part of Figure~\ref{fig:poinc_flow_adapt}. The Schottky surface can then
be thought of being obtained by gluing together two identical domains (see  
Figure~\ref{fig:poinc_flow_adapt}). We will see below that these domains 
correspond to fundamental domains of a McMullen reflection group .

\begin{figure}
\centering
        \includegraphics[width=0.7\textwidth]{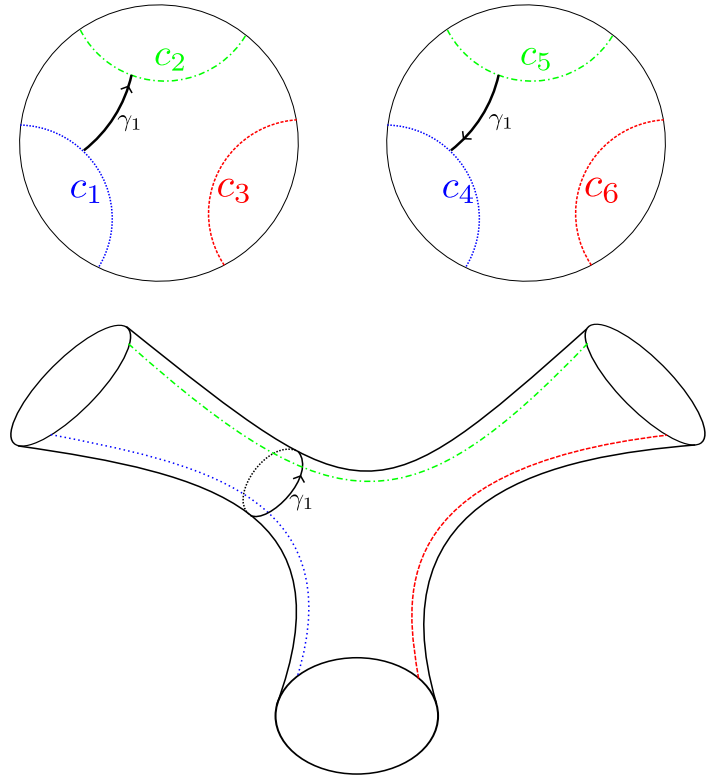}
\caption{Lower part: Schematic sketch of a 3-funneled Schottky 
surface. The red, green and blue lines indicate the cut lines which would correspond
to the Poincaré section of the flow-adapted IFS. The  black line $\gamma_1$
represents a geodesic which winds once around one of the 
funnels. Upper part: Fundamental domain of two McMullen IFS represented in the
Poincaré disk $\mathbb D$. The Schottky surface below can be obtained by gluing 
those two fundamental domains together along the cycles such that the colors match 
each other. In black the two arcs of the geodesic $\gamma_1$ are sketched.}
\label{fig:poinc_flow_adapt}
\end{figure}

The aim of this section is thus to construct a holomorphic IFS 
leading to a dynamical zeta function that also equals the Selberg zeta function 
but which is constructed in the spirit of Figure~\ref{fig:poinc_flow_adapt}. This 
\emph{flow-adapted} IFS will turn out to be the natural choice for proving the 
analyticity of the generalized zeta function 
(Theorem~\ref{thm:gen_zeta_analytic}) and a crucial ingredient for 
proving the geometric limits 
(Theorem~\ref{thm:rescaled_zeta_limit}).

The flow-adapted IFS will be obtained by doubling a McMullen reflection 
IFS \cite{McM98} (see also \cite[Section 6]{JP02}) and we will first recall 
the definition
of a McMullen reflection group. Those groups are best visualized in the 
Poincaré disk model. Let $c_1, c_2, c_3$ denote three geodesics that do not 
intersect. Geometrically these geodesics are circles that are perpendicular to 
the disk boundary $\partial \mathbb D$ (see upper part of 
Figure~\ref{fig:poinc_flow_adapt}). The 
reflection at the geodesic $c_i$ is then an antiholomorphic isometry
\[
 \rho_{c_i}(u):\mathbb D \to \mathbb D
\]
and the Kleinian group $\Gamma$ generated by the reflections $\rho_{c_i}$ is 
called a McMullen reflection group. Note that it contains as well orientation preserving 
(i.e.\,holomorphic) as orientation inverting (i.e.~antiholomorphic) isometries.
The subgroup $\Gamma^+$ of orientation preserving isometries is then a Schottky 
group of a 3-funneled surface containing only hyperbolic transformations. 
If we introduce the displacement length of an hyperbolic positive isometry $T$ as
\[
 l(T):=\min_{u\in \mathbb D} \tu{dist}_{\mathbb D}(u,Tu),
\]
then we can always construct a McMullen reflection group with 
the following properties. 
\begin{lem}\label{lem:theta}
 Let $l_1,l_2,l_3>0$ be real positive numbers, then there
 exist non-intersecting geodesics $c_A,c_B, c_C$ such that
 $\rho_{c_A},\rho_{c_B},\rho_{c_C}$  generate a McMullen reflection 
 group and the displacement length of the composition of two different 
 generators is given by
 \begin{equation}\label{eq:McMullen_displacement_length}
   l(\rho_{c_B}\rho_{c_C})=l_1,~ l(\rho_{c_A} \rho_{c_C})=l_2,~ l(\rho_{c_A}\rho_{c_B})=l_3.
 \end{equation}
\end{lem}
\begin{proof}
 First we use the fact from hyperbolic trigonometry (see e.g. \cite[Lemma 13.2]{Bor07})
 that given three positive numbers
 $\alpha, \beta, \gamma$ there exist positive numbers $A,B,C$ and a right-angled hexagon with 
 side lengths $\alpha,C,\beta,A,\gamma,B$  
 (see Figure~\ref{fig:lengths_of_McMullen:limit}). Note 
 that the geodesic lines $c_A, c_B, c_C$ obtained as the prolongation of $A,B,C$ do not 
 intersect as they are perpendicular to a common geodesic. Thus the reflections along these 
 three geodesics generate a McMullen reflection group. If we choose $\alpha=l_1/2,
 \beta=l_2/2, \gamma=l_3/2$ we also have (\ref{eq:McMullen_displacement_length}) which 
 can be seen as follows. Let $c_\alpha$ be the geodesic prolongation of the side $\alpha$. 
 As it is perpendicular to $c_B$ and $c_C$ it is preserved under the reflection along both 
 circles and is thus also preserved under the hyperbolic element $\rho_{c_B}\rho_{c_A}$. 
 Such an invariant geodesic of an hyperbolic element is also called \emph{axis} and it is 
 known that the displacement length is given for any $u\in c_\alpha$ by (see e.g. 
 \cite[Section 2.1]{Bor07})
 \[
  l(\rho_{c_B}\rho_{c_C})= d_{\mathbb D} (u, \rho_{c_B}\rho_{c_C}u).
 \]
 Choosing $u$ to be the intersection point of $c_C$ and $c_\alpha$ one 
 immediately sees that
 \[
  d_{\mathbb D}(u,\rho_{c_B}\rho_{c_C}u)=d_{\mathbb D}(u,\rho_{c_B}u)=2\alpha=l_1.
 \]

\begin{figure}
\centering
        \includegraphics[width=0.7\textwidth]{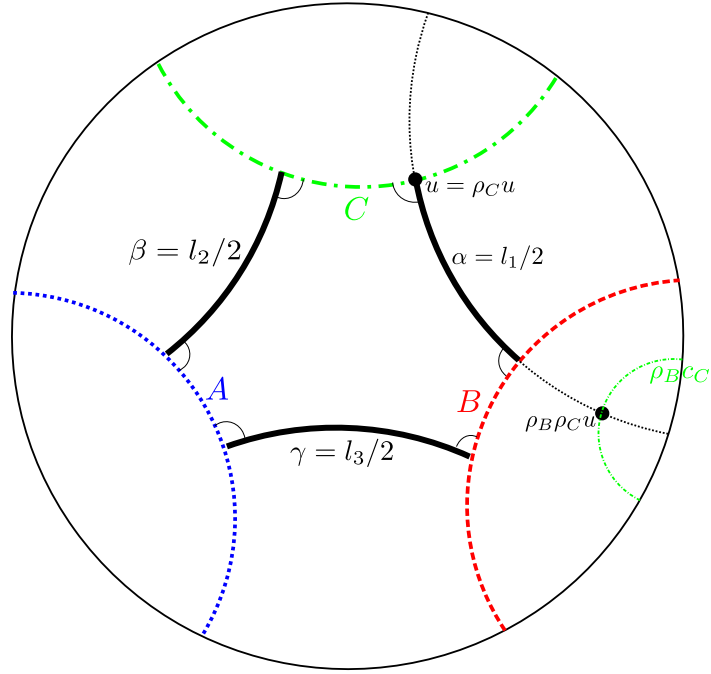}
\caption{Sketch of the orthogonal hexagon in the
Poincaré disk $\mathbb D$ together with the notations from the proof of 
Lemma~\ref{lem:theta}.}
\label{fig:lengths_of_McMullen:limit}
\end{figure}
\end{proof}
The flow-adapted IFS of a Schottky surface $X_{l_1,l_2,l_3}$ will be constructed from the
generators $\rho_{c_A},\rho_{c_B},\rho_{c_C}$. It is however convenient to transform
them by the isometry 
\begin{equation}\label{eq:Cayley}
 C:\Abb{\mathbb D}{\mathbb H}{u}{-i\frac{u-1}{u+1}}
\end{equation}
to the upper half plane. Without loss of generality we can assume that the 
boundary point $-1\in \partial \mathbb D$ is not contained in any of the disks 
bounded by $c_A, c_B, c_C$. The transformation thus gives us 6 points 
$a_1<b_1<a_2<b_2<a_3<b_3 \in \R=\partial \mathbb H$ and three geodesic circles 
$c_i$ with start- and end-points $a_i$ and $b_i$ (see Figure~\ref{fig:FAIFS} for an illustration) . If we denote by $m_i$ the center 
and by $r_i$ the radius of the circle $c_i$ then the reflection at this geodesic
is given by
\[
  \rho_{c_i}(u)=\frac{r}{\overline u - m_i} +m_i
\]
which is an antiholomorphic map on $\mathbb H$. For holomorphic IFS we however need
holomorphic maps on $\C$, so we extend the map antiholomorphically to $\C$ and compose it
with a complex conjugation which gives a holomorphic transformation on $\C$ given by
\begin{equation}\label{eq:R_theta_phi_as_reflection}
 R_{i}(u)=\frac{r_i}{u-m_i}+m_i.
\end{equation}
which can also by expressed as a Moebius transformation with the matrix
\begin{equation} \label{eq:R_as_Moebius}
 R_i = \frac{1}{\sqrt{r_i}}\left(\begin{array}{cc}
                           m_i&r_i-m_i^2\\
                           1&m_i\\
                         \end{array}\right),
\end{equation}

Note that $\det R_i=-1$ thus the matrices $R_i$ are not in $SL(2,\R)$ but any 
product of an even number of $R_i$ is. Finally, by choosing the indices 
appropriately, equation (\ref{eq:McMullen_displacement_length}) transforms to
\begin{equation}\label{eq:R_i_displacement_length}
 l(R_1R_2)=l_1,~l(R_2 R_3)=l_2,~l(R_1R_3)=l_3.
\end{equation}

We can now define the flow-adapted IFS.
\begin{figure}
\centering
        \includegraphics[width=\textwidth]{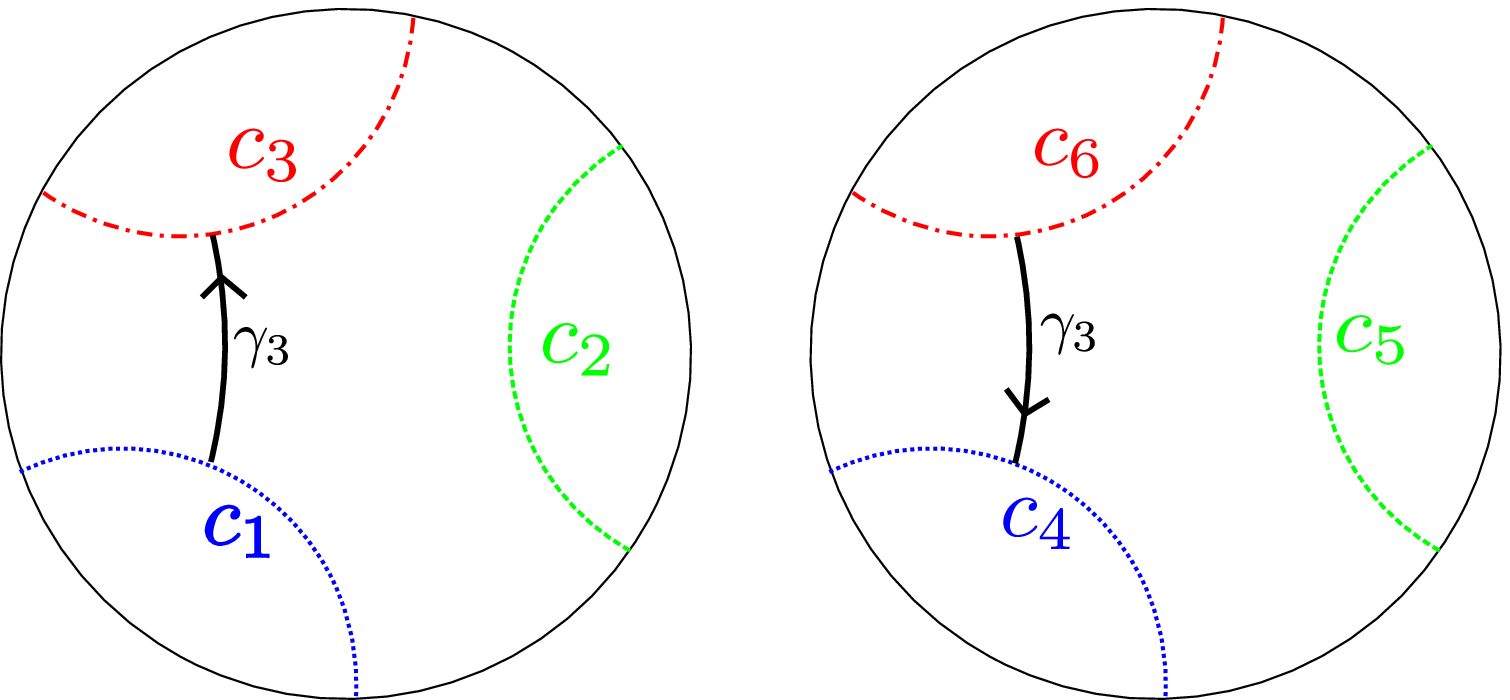}
        \includegraphics[width=\textwidth]{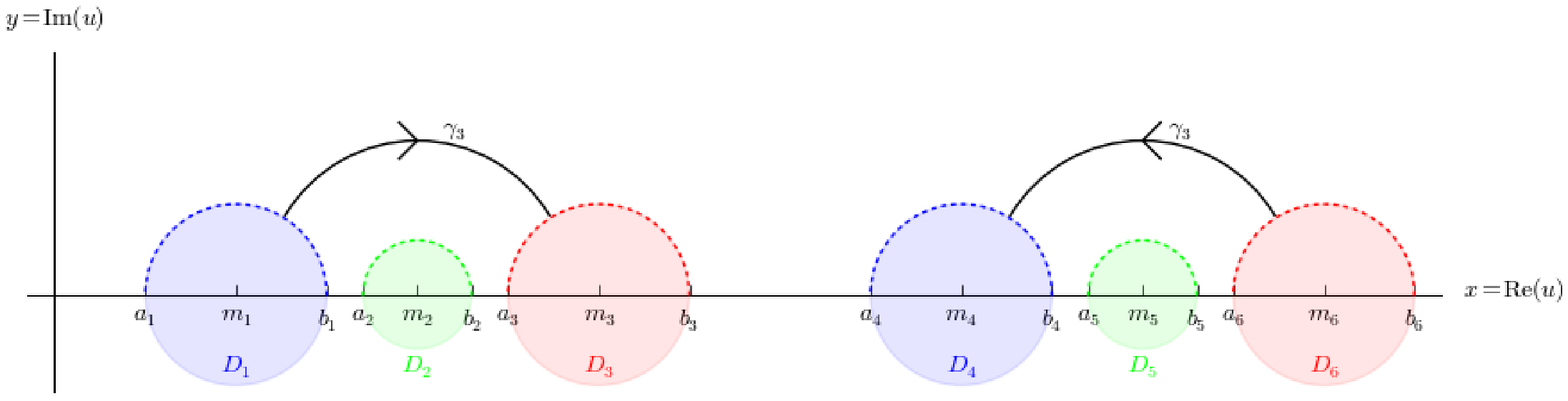}
\caption{Upper part: The two copies
of the fundamental domain of the McMullen reflection group from 
Figure~\ref{fig:poinc_flow_adapt}. Lower part: Disk of the associated
flow-adapted IFS with the notations as in Definition~\ref{def:flow_adapted_IFS:limit}.}
\label{fig:FAIFS}
\end{figure}
\begin{Def}[Flow-adapted IFS]\label{def:flow_adapted_IFS:limit}
 Let $l_1,l_2,l_3$ be real positive numbers and let $c_i$, $a_i$, $b_i$, $m_i$, $r_i$ and 
 $R_i$ be constructed as above from Lemma~\ref{lem:theta}. We define the offset
 variable
 \[
  \delta_{\tu{offset}}:=b_3-a_1 + 1.
 \]
 The \emph{flow-adapted IFS} then is a holomorphic IFS with $N=6$ where the
 disks $D_i$  are the Euclidean disks in $\C$ with centers $m_i$ and radii $r_i$ for $1\leq i
 \leq 3$ and with centers $m_{i-3}+\delta_{\tu{offset}}$ and radii $r_{i-3}$ for 
 $4\leq i \leq 6$. The adjacency matrix $A$ is given by
 $A_{i,j+3}=A_{j+3,i}=1$ for all $1\leq i,j\leq 3$ with $i\neq j$ and $A_{i,j}=0$ else. 
 Finally for $i\rightsquigarrow j$ the maps $\phi_{i,j}$ are given by
 \[
  \phi_{i,j}(u):= \left\{\begin{array}{ll}
                          R_{j-3}(u)+\delta_{\tu{offset}} &\tu{for } i\leq 3\\
                          R_j(u-\delta_{\tu{offset}}) &\tu{for } i>3.\\
                         \end{array}
  \right.
 \]
\end{Def}
\begin{rem}\label{rem:flow_adapted_IFS}
 Note that the concrete form of the flow-adapted IFS is far from being uniquely
 defined by the lengths $l_i$. The three lengths only determine uniquely the 
 side lengths of the orthogonal hexagon in the proof of Lemma~\ref{lem:theta}
 but not its orientation and position inside $\mathbb D$. So every other 
 realization of this pentagon where the point $-1$ is not contained in any of the 
 disks leads to an equivalent IFS. Additionally the offset variable is completely 
 arbitrary, provided it assures that the disks $D_i$ are mutually disjoint. 
\end{rem}
As indicated in the discussion above, we want to show that the dynamical zeta 
function of the flow-adapted IFS with a suitable potential also equals the 
Selberg zeta function. The key ingredient for this equality is, as in the case of 
the ordinary Bowen-Series map, a one-to-one correspondence between the 
classes of prime closed words of the IFS and the primitive closed geodesics 
which we want to state and prove next.
\begin{prop}\label{prop:SymIFS_orbit_geodesic_equiv:limit}
 Let $l_1,l_2,l_3$ be positive, real numbers and consider the corresponding 
 flow-adapted IFS from Definition~\ref{def:flow_adapted_IFS:limit}. 
 Then there exists a bijection between
 the classes of prime words in $\left[ \mathcal W^{\mathrm{prime}}\right]$ and 
 the primitive closed
 geodesics on $X_{l_1,l_2,l_3}$. Additionally the  length of the geodesic 
 associated to $[w]$ is given by
 \begin{equation}\label{eq:length_of_geod:limit}
   -\log(\phi_w'(u_w)).
 \end{equation}
\end{prop}
\begin{proof}
Let $R_{1},R_{2},R_{3}$ be as in Definition~\ref{def:flow_adapted_IFS:limit} and
$\Gamma=\langle R_{1},R_{2},R_3\rangle^+$ the subgroup 
of orientation preserving isometries of the McMullen reflection group.
Then $\Gamma$ is generated by the two hyperbolic isometries
$S_1=R_1R_2$ and $S_2=R_2R_3$ and it is straightforward to check that $S_1,S_2$
generate the Schottky group $\Gamma_{l_1,l_2,l_3}$. Thus it is known (see 
e.g.~\cite[Proposition 2.16]{Bor07}) that the set of primitive closed geodesics 
on $X_{l_1,l_2,l_3}$ is in bijection to 
the set of primitive conjugacy classes $[T]\in\Gamma$ where primitive means that
there is no $S\in [T]$ such that $S=R^k$ for some $R\in \Gamma$ and $k>1$. 
Consequently our aim is to construct a bijection
\[
 T:\left[ \mathcal W^{\mathrm{prime}}\right] \to \left\{\tu{primitive conjugacy classes of }\Gamma\right\}.
\]
In order to do so, we note that for $w\in \mathcal W_k$ from the form of 
the adjacency 
matrix in Definition~\ref{def:flow_adapted_IFS:limit} we have 
$w_i\leq 3 \Rightarrow w_{i+1}>3$. 
Thus, if $w$ is a closed word, $k$ has to be even. We first define the map
\[
 T:\left[ \mathcal W^{\mathrm{cl}}\right] \to \left\{\tu{conjugacy classes of }\Gamma\right\}.
\]
on the closed words and will later show that we can easily restrict it to the
prime words. For a closed  word  $w=(w_0,\ldots,w_{2r})$ we define the 
map $T$ by
\[
 T(w):= \left\{\begin{array}{ll}
                          R_{w_{2r}}R_{w_{2r-1}-3}\ldots R_{w_2}R_{w_1 -3}  &\tu{if } w_0\leq 3 \\
                          R_{w_{2r-1}}R_{w_{2r-2}-3} \ldots R_{w_1}R_{w_0-3}  &\tu{if } w_0> 3 
                         \end{array}.
  \right. 
\]
As closed words have to be of even length, $T(w)$ consists of an even number of 
reflections and is thus a positive isometry. We first show that $T$ is well 
defined on $\left[\mathcal W^{\mathrm{prime}}\right]$, i.e.\,that it doesn't 
depend on the choice of
the representative of $[w]$. So let $v\in [w]$. Without loss of generality we can 
assume that $w_0\leq 3$ and $v_0\leq 3$ because otherwise we could simply apply
the right-shift $\sigma_R$ to obtain such an element in the same equivalence
class that fulfills this condition and that is mapped to the identical element in
$\Gamma$. Consequently there exists an integer $0\leq t\leq r$ such that 
$v=(w_{2t},\ldots,w_{2r},w_1\ldots,w_{2t})$ and we obtain
\[
 T(v)=R_{w_{2t}}\ldots R_{w_1-3}R_{w_{2r}}\ldots R_{w_{2t+2}}R_{w_{2t+1}-3} =S^{-1}T(w)S
\]
for $S=R_{w_{2r}}\ldots R_{w_{2t+1}-3}$. Thus $T(v)$ 
is in the same conjugacy class as $T(w)$.

In order to see the injectivity we take two words $v$ and $w$ that are mapped to 
the same conjugacy class. We assume first that 
\[
 T(v)=R_aR_b T(w)R_bR_a.
\]
However from the form of the adjacency matrix, it is not possible that an element in
the image of $T$ starts and ends with the same generator. Thus we have either
\[
R_bR_a=R_{w_1-3}R_{w_2}
\]
or 
\[
R_aR_b=R_{w_{2r-1}-3}R_{w_{2r}}.
\]
In the first case we have $v=\sigma_L^2 w$ in the latter case $v=\sigma_R^2 w$.
By iterating this argument for arbitrary conjugations of $T(w)$ and $T(v)$ we have
shown the injectivity of the map $T$.

In order to see the surjectivity, let $S\in \Gamma$ be an arbitrary element. By
definition of $\Gamma$ we can write $S=R_{s_{2r}}\ldots R_{s_1}$ with 
$1\leq s_i\leq 3$. As two consecutive
identical reflections cancel each other we can assume that $s_{i}\neq s_{i+1}$. 
Finally while $s_{1}=s_{2_r}$ we can conjugate $S$ by $R_{s_2}R_{s_1}$ which leads 
to an element composed from $2r-2$ reflections. By iterative conjugation we can thus
reduce the element to $\tilde S=R_{\tilde s_{2\tilde r}}\ldots R_{\tilde s_1}$ with 
$\tilde s_{1}\neq\tilde s_{2\tilde r}$ and we obtain
\[
 \tilde S=T((s_{2\tilde r},s_1+3,s_2,\ldots, s_{2\tilde r-1}+3,s_{2\tilde r})).
\]

We have thus constructed a bijective map between the classes of closed words 
and the conjugacy classes in $\Gamma$. We will now prove that this map can
be restricted to a bijection between the classes of prime words and the 
primitive conjugacy classes. As $T$ is a bijectiona and on both sides an element can either be 
primitive or composite it suffices to show that $T$ maps composite closed words
to composite conjugacy classes. This is, however, straight forward from the 
definition of $T$ as obviously $T([w^k])=T([w])^k$. 

With this restriction we have constructed a bijection between the classes
of closed, prime words and primitive conjugacy classes. Using the above
mentioned result on the one-to-one correspondence between oriented primitive 
geodesics and primitive conjugacy classes, this is equivalently a bijection 
to the set of primitive, 
oriented, closed geodesics and it only remains to prove (\ref{eq:length_of_geod:limit}). 

In order to achieve this, we first recall that the length of the primitive geodesic 
associated to a conjugacy class of an hyperbolic element $T\in \Gamma$ is equal to 
the displacement length of $T$ (see e.g. \cite[Proposition 2.16]{Bor07})
and it is also a well known fact that if $u_T\in\partial \mathbb H$ is the stable 
fixed point of $T$ then $l(T) =-\log((T)'(u_T))$ (see e.g. \cite[(15.2)]{Bor07}).
Next we recall from the proof of Theorem \ref{thm:dynamical_Selberg_equiv} that
$\phi_w'(u_w)$ is independent of the representative in $[w]$. Assuming once more, 
that $w_0 \leq 3$ we calculate that
\[
 \phi_w(u_w)=R_{w_{2r}}\ldots R_{w_1-3} u_w.
\]
Thus $u_w$ is the stable fixed point of the hyperbolic element $T(w)$ and 
for the displacement length of $T$ we obtain $l(T(w)) =-\log((T(w))'(u_w))$. 
As however the displacement length coincides with the 
length of the associated closed geodesic (see e.g. \cite[Proposition 2.16]{Bor07})
we established (\ref{eq:length_of_geod:limit}) and finished the proof of 
Proposition~\ref{prop:SymIFS_orbit_geodesic_equiv:limit}.
\end{proof}
\begin{cor}
 Let $l_1,l_2,l_3$ be real positive numbers, and $\mathcal L_s$ the Ruelle
 transfer operator of the flow-adapted IFS as defined in
 Definition~\ref{def:flow_adapted_IFS:limit} with potential $V_s(u)=[(\phi^{-1})'(u)]^{-s}$,
 then the dynamical zeta function coincides with the Selberg zeta function of 
 $X_{l_1,l_2,l_3}$
 \[
  Z_{X_{l_1,l_2,l_3}}(s)=\det(1-\mathcal L_s)
 \]
\end{cor}
\begin{proof}
 As (\ref{eq:zeta_product_form}) did not depend on the choice of the IFS we 
 obtain also for the flow-adapted IFS
 \[
  \det(1-\mathcal L_s)=
                       \prod\limits_{[w]\in\left[\mathcal W^{\tu{prime}}\right]}
                       \prod\limits_{k\geq 0}\left(1-\phi_w'(u_w)^{k+s}\right).
 \]
 Using Proposition~\ref{prop:SymIFS_orbit_geodesic_equiv:limit} this can be written as
 \[
  \det(1-\mathcal L_s)=\prod\limits_{\gamma\in \mathcal P_{X_{l_1,l_2,l_3}}}
                       \prod\limits_{k\geq 0}                       
                       \left(1-e^{-(k+s)l(\gamma)}\right).
 \]
which is exactly the Selberg zeta function of $X_{l_1,l_2,l_3}$. 
\end{proof}
With help of the flow-adapted IFS we can now prove 
the analyticity of the generalized zeta functions which was stated in the introduction
as Theorem~\ref{thm:gen_zeta_analytic}.
\begin{customthm}{\bf \ref{thm:gen_zeta_analytic}}
\emph{ Let $X_{l_1,l_2,l_3}$ be a Schottky surface with three funnels of widths 
 $l_1,l_2,l_3$ and let $n_1,n_2,n_3\in \N$. We define 
 \[
\mathbf{n}:\Abb{\mathcal P_{X_{l_1,l_2,l_3}}}{\N}{\gamma}{\sum_{i=1}^3 n_iw_i(\gamma)}  
 \]
where $w_i(\gamma)$ denotes the winding number around the funnel of width $l_i$. 
Then the generalized zeta function
\[
 d_\mathbf{n}(s,z) = \prod\limits_{\gamma\in \mathcal P_{X_{l_1,l_2,l_3}}}
                       \prod\limits_{k\geq0}                       
                       \left(1-z^{\mathbf{n}(\gamma)}e^{-(k+s)l(\gamma)}\right).
\]
extends to an analytic function on $\C^2$.}
\end{customthm}
\begin{proof}
First we note that for $|z|<1$ and $\tu{Re}(s)>1$  the products in (\ref{eq:gen_zeta:limit}) expand to an absolutely convergent series.
In the region of absolute convergence we can Taylor expand 
$d_{\mathbf{n}}$ in $z$ around zero and obtain
\begin{equation}\label{eq:taylor_gen_zeta}
   d_{\mathbf n}(s,z)=\sum_{k=0}^{\infty} b_k(s)z^k.
\end{equation}
 In order to show the analytic continuation we construct an appropriate $s$- 
 and $z$-dependent trace class operator.  We take the flow-adapted IFS and 
 define a potential $V$ which depends 
 analytically on two complex parameters $s,z\in \C$ by setting 
 for $i\rightsquigarrow j$ and $u\in\phi_{i,j}(D)$
 \[
  V(u;s,z):=z^{n_{i,j}}[-(\phi^{-1})'(u)]^{-s}\\
 \]
 where 
 \begin{eqnarray*}
   n_{1,5}=n_{5,1}=n_{4,2}=n_{2,4}&:=& n_1\\
   n_{2,6}=n_{6,2}=n_{5,3}=n_{3,5}&:=& n_2\\
   n_{3,4}=n_{4,3}=n_{6,1}=n_{1,6}&:=& n_3
 \end{eqnarray*}
 Note that $-(\phi^{-1})'$ is non-vanishing on $\phi(D)$ and real and positive 
 on $\Phi(D)\cap \R$, so we can extend for each $s\in \C$ the function
 $[-(\phi^{-1})']^{-s}$ from 
 the real line to each of the disks $\phi_{i,j}(D_i)$ and obtain this way a 
 family of holomorphic and bounded potentials on $\phi(D)$ that depends
 analytically on $s$ and $z$. Following \cite[Proposition 2]{JP02} 
 the family of transfer operators 
\begin{equation}\label{eq:flow_adapted_transfer_op}
\mathcal L_{s,z}:=\mathcal L_{V(\bullet;s,z)} 
\end{equation}
with this potential is nuclear on $\mathcal A_\infty(D)$ and  as a 
consequence of the analytic 
dependence of $V$ on the parameters $s,z$ the Fredholm determinant 
$\det(1-\mathcal L_{s,z})$ also is an analytic function of $s,z$. 
The choice of the factors $z^{n_i}$ is exactly such that each 
half winding around one of the $i$-th funnel contributes with a factor 
$z^{n_i}$. Thus each winding around the $i$-th funnel contributes with $2n_i$ and
the total dynamical zeta function is in the region of 
absolute convergence given by
\begin{equation}\label{eq:gen_zeta:limit_tilde}
 \tilde d_{\mathbf n}(s,z):=\det(1-\mathcal L_{s,z})= \prod\limits_{\gamma\in \mathcal P_{X_{l_1,l_2,l_3}}}
                       \prod\limits_{k\geq0}                       
                       \left(1-z^{2\mathbf{n}(\gamma)}e^{-(k+s)l(\gamma)}\right).
\end{equation}
As we know  
that the function is analytic we can Taylor expand it in $z$ around $0$ and obtain 
\begin{equation}\label{eq:taylor_gen_zeta_tilde}
   d_{\mathbf n}(s,z)=\sum_{k=0}^{\infty} b_k(s)z^k.
\end{equation}
with analytic coefficients $b_k(s)$. As in (\ref{eq:gen_zeta:limit_tilde}) only even 
powers of $z$ appear we immediately can conclude $b_{2k+1}(s)=0$. Comparing 
furthermore the product expressions (\ref{eq:gen_zeta:limit}) with 
(\ref{eq:gen_zeta:limit_tilde}) and the Taylor expansions (\ref{eq:taylor_gen_zeta}) 
with (\ref{eq:taylor_gen_zeta_tilde}) in the region of absolute convergence 
we obtain
\[
 b_{k}(s) = \tilde b_{2k}(s).
\]
By this identification we obtain an analytic continuation of the Taylor 
coefficients $b_k(s)$. As for each $s\in \C$ the power series 
(\ref{eq:taylor_gen_zeta_tilde}) has a radius of convergence equal to 
infinity, i.e.\,$\limsup_{k}|\tilde b_k(s)|^{1/k} =0$ we also obtain 
that (\ref{eq:taylor_gen_zeta}) converges for all $z\in \C$ and the 
generalized zeta function is thus analytic. 
\end{proof}
\section{Geometric limits}\label{sec:geom_limit}
In this section we will prove Theorem~\ref{thm:rescaled_zeta_limit} and then 
show that Theorem~\ref{thm:location_res} is a consequence of this result. 
The proof of Theorem~\ref{thm:rescaled_zeta_limit} will be performed in 
three steps: First we will derive a 
form of the flow-adapted IFS, that is especially suited to treat the family of Schottky surfaces 
in the limit $\ell\to\infty$. In a next step (Lemma~\ref{lem:zeta_trunctaion}) 
we will derive explicit bounds for 
the coefficients of the cycle expansion of the generalized zeta function. Finally
we will be able to prove the convergence using these bounds and the special
form of the flow-adapted IFS (See Lemma~\ref{lem:convergence_of_leading_cycle} and 
\ref{lem:limit_of_potential} as well as the rest of this section).

We start with the construction of the special form of flow-adapted IFS. 
\begin{lem}\label{lem:family_flow_IFS}
 Let  $n_1,n_2,n_3$ be positive integers fulfilling the triangle condition. 
 Then there exists $\ell_0$ such that for any $\ell>\ell_0$ there exists a 
 family of flow-adapted IFS associated to $X_{n_1,n_2,n_3}(\ell)$ in the sense of
 Definition~\ref{def:flow_adapted_IFS:limit} such
 that the lower boundaries $a_j$ of the disks $D_j$ are given by $a_j=2(j-1)$ independently 
 of $\ell$ and $r_j<0.5$. Furthermore the radii fulfill 
 the asymptotics
 \begin{equation}\label{eq:asymptotic_r_j}
  \lim\limits_{\ell\to\infty} r_j(\ell)e^{\kappa_j \ell} = C_j
 \end{equation}
 where
 \[
  \kappa_1= \frac{n_1+n_3-n_2}{2},~\kappa_2=\frac{n_1+n_2-n_3}{2},~
  \kappa_3=\frac{n_2+n_3-n_1}{2}
 \]
 and $\kappa_4=\kappa_1,~\kappa_5=\kappa_2,~\kappa_6=\kappa_3$ are constants strictly larger then zero and
 \[
  C_1=C_3=C_4=C_6=8,~C_2=C_5=\frac{1}{2}.
 \]
\end{lem}
\begin{proof}
 We will first use the freedom of choosing the position and orientation of the 
 hexagon as already mentioned in Remark~\ref{rem:flow_adapted_IFS}. Instead of 
 constructing the reflection group on $\mathbb D$ we can also directly work on 
 the upper half plane (see Figure~\ref{fig:construct_family_IFS}). So we can 
 consider again an orthogonal hexagon with 
 side lengths $A,n_1\ell /2, B, n_2\ell /2, C, n_3\ell /2$ and we call
 $c_1$ to be the geodesic prolongation of the side $A$, $c_2$ the one of $B$ and 
 $c_3$ the one of $C$. Now there exists an isometry such that for the starting 
 points $a_j\in \R=\partial \mathbb H$ of $c_j$ we have $a_j=2j$. This isometry 
 can be constructed in three steps: First translate parallel to the real axis until $a_1=0$,
 then apply the dilation $z\to\lambda z$ which fixes $a_1$ until $a_2=2$ and finally 
 apply the one parameter group of hyperbolic transformation that fixes $a_1, a_2$
 until $a_3=6$ is fulfilled. Setting the offset parameter $\delta_{\tu{offset}}=6$ we
 then obtain the condition $a_j=2(j-1)$ for all $1\leq j\leq 6$. Note however that the 
 flow-adapted IFS might be ill defined with this offset parameter, because
 we could in principle have $r_3>1$. In the next step we will show however that in the
 limit $\ell \to\infty$ all radii $r_j$ will tend to zero. Thus for sufficiently
 large $\ell$ everything is well defined. 
 
\begin{figure}
\centering
        \includegraphics[width=1.1\textwidth]{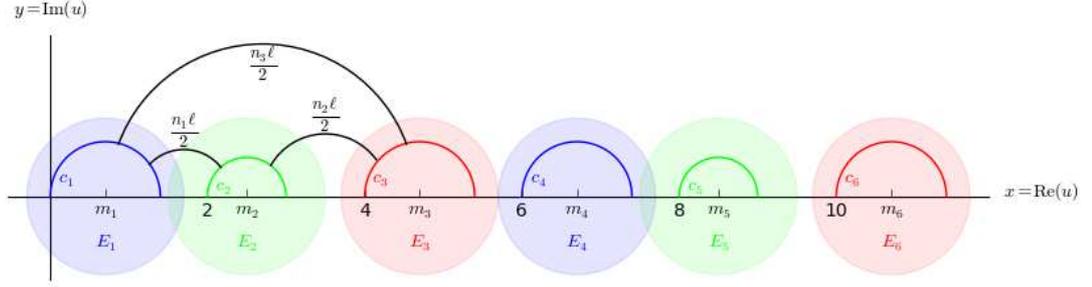}
\caption{Illustration of the construction of the family of flow-adapted IFS 
in Lemma \ref{lem:family_flow_IFS}. The start points of the circles $c_j$
are now fixed to $2(j-1)$. The light colored disks indicate the extended
disks $E_j$ which are crucial for obtaining the estimates in 
Lemma~\ref{lem:zeta_trunctaion}.
}
\label{fig:construct_family_IFS}
\end{figure}
 In order to show the convergence of the radii to zero we  first note that 
 even without the triangle condition at least two of the radii 
 have to converge towards zero. Otherwise the perpendicular distance between those
 two circles can never tend towards infinity as already the distance between their 
 start points $a_i$ is fixed. We can thus assume, after possibly permuting the $l_i$
 that $r_1$ and $r_3$ converge to zero. For a proof by contradiction we now assume that
 $r_2$ is bounded away from zero by $r_{\tu{min}}$. We will first show that then
 also the side length $B$ is bounded away from zero: For $x\in\partial \mathbb H$ and
 $r>0$ we consider the unique geodesic that starts in $x$ and is orthogonal to 
 the circle of radius $r$ that starts at $a_2$. We denote the intersection point 
 of these two geodesics with $p(x,r)$. Then for two different points $x_1\neq x_2$ 
 the points $p(x_1,r)$ and $p(x_2,r)$ are different. Recall that $B$ is exactly the hyperbolic 
 distance between $p(r_2,x_1)$ and $p(r_2,x_2)$ where $x_1\in [a_1,a_1+2 r_1]$ and 
 $x_2\in [a_3,a_3+2r_3]$ such that the geodesics are also orthogonal to $c_1$ and $c_3$,
 respectively (see Figure~\ref{fig:lower_bound_on_B} for an illustration of these points).
 From the fact that the disks $D_i$ are mutually disjoint we conclude, 
 that $[a_1,a_1+2r_1]\cap[a_3,a_3+2r_3]=\emptyset$ so 
 $d_{\mathbb H}(p(x_1,r_2),p(x_2,r_2))>0$ for all $\ell$. Furthermore the disjoint disk 
 and the lower bound on $r_2$ together imply that $r_2\in[r_{\tu{min}},1]$. The fact that
 $r_1$ and $r_2$ converge to zero finally means that there exist $r_{1,\max}, r_{3,\max}$
 such that $r_1\leq r_{1,\max}$ and $r_3\leq r_{3,\max}$ for all $\ell$. We can thus bound
 \[
  B=d_{\mathbb H}(p(x_1,r_2),p(x_2,r_2))\geq \min\limits_{\begin{array}{c}
                                                          y_1\in[a_1,a_1+2r_{1,\max}],\\ y_2\in[a_3,a_3+2r_{3,\max}],\\
                                                          r\in [r_{\min},1]
                                                          \end{array}
                                                          }
                                                          d_{\mathbb H}(p(y_1,r),p(y_2,r))=:B_0.
 \] 
\begin{figure}
\centering
        \includegraphics[width=1.1\textwidth]{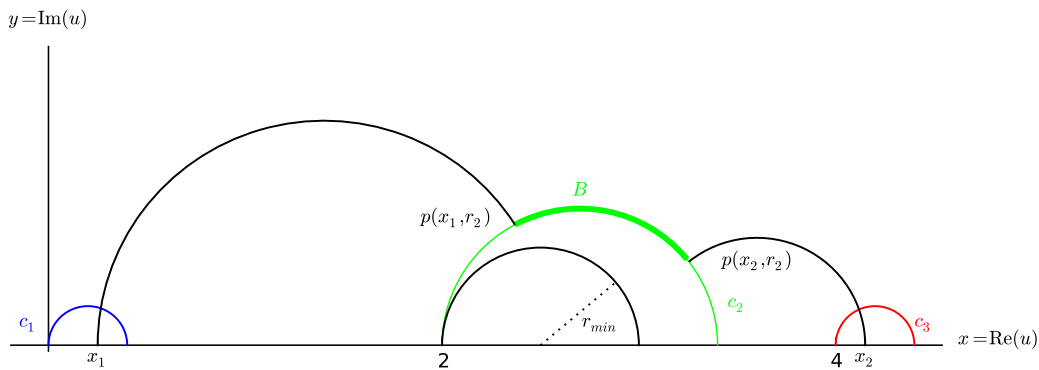}
\caption{Illustration of the notation for the lower bound on 
$B$ in the proof of Lemma~\ref{lem:family_flow_IFS}.
}
\label{fig:lower_bound_on_B}
\end{figure}
 As $d_{\mathbb H}(p(y_1,r),p(y_2,r))$ is a positive quantity that depends continuously on 
 the parameters $r,y_1,y_2$ which vary in a compact set, $B_0>0$ and $B$ is bounded away 
 from zero. This is however in contradiction to the triangle condition. From \cite[(2.6.10)]
 {Thu02} we have the formula for orthogonal hexagons
 \[
  \cosh B = \frac{\cosh (n_1\ell /2) \cosh (n_2\ell /2)  + \cosh (n_3\ell /2)}
  {\sinh(n_1 \ell/2)\sinh(n_2 \ell /3)}.
 \]
 In the limit $\ell\to \infty$ the right side becomes
 \[
  1+\frac{e^{n_3\ell}}{e^{(n_1+n_2)\ell}}
 \]
 which converges to 1 if the triangle condition is fulfilled and consequently $\lim_{\ell\to\infty} B=0$. We have thus
 shown that under the triangle condition all three radii have to converge to zero. 
 
 In order to prove the concrete form of the asymptotics (\ref{eq:asymptotic_r_j}) we use
 the following general formula for the displacement length of an hyperbolic 
 element $T\in SL(2,\R)$
 \[
  \cosh(l(T)/2) = |\Tr(T)/2|.
 \]
 So using (\ref{eq:R_i_displacement_length}) together with the explicit form 
 (\ref{eq:R_as_Moebius}) for the $R_i$ one obtains the set of equations. 
 \begin{eqnarray*}
  \cosh(n_1\ell/2) =\cosh(l(R_1R_2)/2) = \left|\frac{\Tr(R_1R_2)}{2}\right| &=& \frac{(m_1-m_2)^2 -r_1-r_2}{2\sqrt{r_1 r_2}}\\
  \cosh(n_2\ell/2) =\cosh(l(R_2R_3)/2) = \left|\frac{\Tr(R_2R_3)}{2}\right| &=& \frac{(m_2-m_3)^2 -r_2-r_3}{2\sqrt{r_2 r_3}}\\
  \cosh(n_3\ell/2) =\cosh(l(R_1R_3)/2) = \left|\frac{\Tr(R_1R_3)}{2}\right| &=& \frac{(m_1-m_3)^2 -r_1-r_3}{2\sqrt{r_1 r_3}}.
 \end{eqnarray*}
Dividing both sides by $e^{n_i\ell/2}$ and taking the limit $\ell\to\infty$
one obtains
\begin{eqnarray*}
 1=\lim\limits_{\ell\to\infty} \frac{\cosh{n_1\ell/2}}{e^{n_1\ell/2}}&=&
   \lim\limits_{\ell\to\infty}\frac{(m_1-m_2)^2 -r_1-r_2}{2\sqrt{r_1 r_2}e^{n_1\ell/2}}\\
 1=\lim\limits_{\ell\to\infty} \frac{\cosh{n_2\ell/2}}{e^{n_1\ell/2}}&=&
   \lim\limits_{\ell\to\infty}\frac{(m_2-m_3)^2 -r_2-r_3}{2\sqrt{r_2 r_3}e^{n_2\ell/2}}\\
 1=\lim\limits_{\ell\to\infty} \frac{\cosh{n_3\ell/2}}{e^{n_3\ell/2}}&=&
   \lim\limits_{\ell\to\infty}\frac{(m_1-m_3)^2 -r_1-r_3}{2\sqrt{r_1 r_3}e^{n_3\ell/2}}.
\end{eqnarray*}
From the fact that the $a_i$ do not depend on $\ell$ and that the radii all 
converge to zero we explicitly know 
$\lim\limits_{\ell\to\infty} (m_i-m_j)^2-r_i-r_j$ and obtain
\begin{eqnarray*}
 \lim\limits_{\ell\to\infty} \sqrt{r_1 r_2}e^{n_1\ell/2} &=& 2\\
 \lim\limits_{\ell\to\infty} \sqrt{r_2 r_3}e^{n_2\ell/2} &=& 2\\
 \lim\limits_{\ell\to\infty} \sqrt{r_1 r_3}e^{n_3\ell/2} &=& 8.\\
\end{eqnarray*}
We now multiply two of these equations and divide by the third one and obtain
\begin{eqnarray*}
 \lim\limits_{\ell\to\infty} r_1e^{ (n_1+n_3-n_2)\ell/2} &=& 8\\
 \lim\limits_{\ell\to\infty} r_2e^{(n_1+n_2 -n_3)\ell/2} &=& \frac{1}{2}\\
 \lim\limits_{\ell\to\infty} r_3e^{(n_2+n_3-n_1)\ell/2} &=& 8\\
\end{eqnarray*}
which finishes the proof of Lemma~\ref{lem:family_flow_IFS}.
\end{proof}
From the property $r_j<0.5$ of these flow-adapted IFS it directly follows that
for any $\ell>\ell_0$ and any $1\leq j\leq 6$ there exists extended disks $E_j$ 
which are concentric with $D_j$, have a radius $r_{E_j}>1$ and do not intersect
any of the other disks $D_i$ (see Figure~\ref{fig:construct_family_IFS}) 
for an illustration).
\begin{lem}\label{lem:zeta_trunctaion}
 Let $n_1,n_2,n_3$ be positive integers fulfilling the triangle condition and 
 $\ell>\ell_0$ as in 
 Lemma~\ref{lem:family_flow_IFS}. Let $d_{\mathbf n}(s,z)$ be the generalized 
 zeta function of the family of Schottky surfaces $X_{n_1,n_2,n_3}(\ell)$. Let 
 furthermore $\mathcal L_{s,z;\ell}$ be the transfer operator as defined in
 (\ref{eq:flow_adapted_transfer_op}) of the flow-adapted IFS from 
 Lemma~\ref{lem:family_flow_IFS} and let its cycle expansion be given by
 \begin{equation}
   \det(1-y\mathcal L_{s,z;\ell})=1+\sum\limits_{k=1}^{\infty} y^k\tilde d_{\mathbf n}^{(k)}(s,z;\ell).
 \end{equation}
With this definition of $\tilde d_{\mathbf{n}}^{(k)}(s,z;\ell)$ we can express 
the generalized zeta function as
\begin{equation}\label{eq:gen_zeta:limit_cycle}
 d_{\mathbf{n}}(s,z;\ell)=1+\sum\limits_{k=1}^{\infty} \tilde d_{\mathbf n}^{(k)}(s,\sqrt{z};\ell).
\end{equation}
If we furthermore fix six disks  $E_j$ with radius $r_{E_j}>1$ and center $m_j$ 
such that for $i\neq j$ $E_i\cap D_j=\emptyset$ then we have the explicit bounds
\begin{equation}\label{eq:cycle_coeff_bound}
 |\tilde d_{\mathbf{n}}^{(k)}(s,z;\ell)| \leq  k^{k/2} K(s,z;\ell)^k 
   \sum_{m_1<\dots<m_k} r(\ell)^{\lfloor m_1/6 \rfloor+\ldots+ \lfloor m_k/6\rfloor}
 ,
\end{equation}
where $\lfloor x\rfloor$ denotes the integer part of a real number $x$ and 
\begin{equation}\label{eq:gamma_l_K_szl}
 r(\ell) :=\max\limits_{1\leq j\leq 6} r_j~~ \tu{ and }~  K(s,z;\ell):= \frac{1}{2\pi}\max_{j\rightsquigarrow i}\|V(\phi_{j,i}(\bullet);s,z)\|_{\infty, E_j}.
\end{equation}

\end{lem}
\begin{rem}
 We will follow closely the techniques of Jenkinson-Pollicott \cite{JP02} which 
 they used to obtain rigorous dimension estimates, but in a different spirit. 
 While they were interested in the question how fast the cycle expansion 
 coefficients $d^{(k)}$ for a fixed IFS vanish if the order $k$ becomes large, 
 we consider the coefficients for a family of IFS at a fixed order $k$ and we 
 want to determine which coefficients vanish in the limit $\ell\to\infty$. 
 The estimate (\ref{eq:cycle_coeff_bound}) then says that only the first six
 coefficients survive in this limit.
\end{rem}
\begin{proof}
 We know that the transfer operator $\mathcal L_{s,z;\ell}$ of the flow-adapted 
 IFS is a nuclear operator, so using the following result of Grothendieck we 
 obtain a direct formula for the coefficients of the cycle expansion in terms 
 of the nuclear representation.
 \begin{prop}[Grothendieck 1956 \cite{Gro56}]\label{prop:Grothendieck_taylor_coeff}
  If $\mathcal B$ is a Banach space and $\mathcal L:\mathcal B \to\mathcal B$ 
  is a nuclear operator with the nuclear representation 
  $\mathcal L h =\sum\limits_{n=1}^\infty \lambda_n \alpha_n(h)v_n$ as defined 
  in Definition~\ref{def:nuclear}, then the 
  Fredholm determinant of $\mathcal L$ can be expanded in a power series 
  $\det(1-z\mathcal L)=1+\sum_{k=1}^\infty z^k d^{(k)}$ with 
  \begin{equation}\label{eq:Grothendieck_taylor_coeff}
    d^{(k)} = (-1)^k\sum_{m_1<\dots<m_k} \lambda_{m_1}\ldots\lambda_{m_k} \det\left[(\alpha_{m_p}(v_{m_q}))_{p,q=1}^k\right]
  \end{equation}
  where 
  \[
   (\alpha_{m_p}(v_{m_q}))_{p,q=1}^k = \left(
                                   \begin{array}{ccc}
                                    \alpha_{m_1}(v_{m_1}) &\dots& \alpha_{m_k}(v_{m_1})\\
                                    \vdots &\ddots& \vdots\\
                                    \alpha_{m_1}(v_{m_k}) &\dots&\alpha_{m_k}(v_{m_k})
                                   \end{array}
\right)
  \]
is a $k\times k$ matrix.
\end{prop}
Thus (\ref{eq:Grothendieck_taylor_coeff}) allows us to obtain estimates on 
the coefficients $\tilde d_{\mathbf{n}}^{(k)}(s,z;l)$ in terms of the 
nuclear representation of $\mathcal L_{s,z;\ell}$. We therefore want to 
derive its nuclear decomposition now and obtain explicit estimates on the 
$\lambda_n$.

Recall from (\ref{eq:transfer_op_sum_decomposition}) that 
$\mathcal L_{s,z;l}$ can be decomposed into a sum of the following 
six operators
\[
 \left(\mathcal L^{(j)} h\right)(u) = \sum\limits_{i \tu{ s.t.} j\rightsquigarrow i} V(\phi_{j,i}(u)) h(\phi_{j,i}(u)).
\]
with $1\leq j\leq 6$.

It is now an important remark that the function $\mathcal L_{s,z;\ell}^{(j)} h$ 
is not only holomorphic on $D_j$ but can be extended holomorphically to a 
much larger disk with the same center. We illustrate the mechanism 
for $\mathcal L_{s,z;\ell}^{(1)} h$ which is given by
\begin{eqnarray*}
 (\mathcal L_{s,z;\ell}^{(1)} h)(u) &=& V(\phi_{1,5}(u);s,z)h(\phi_{1,5}(u)) + V(\phi_{1,6}(u);s,z)h(\phi_{1,6}(u))\\
 &=&z^{n_1}\left[-R_2'(u)\right]^{s} h(R_2(u)+\delta_{\tu{offset}}) + z^{n_3}\left[-R_3'(u)\right]^{s} h(R_3(u)+\delta_{\tu{offset}}).
\end{eqnarray*}
The factor  $h(R_2(\bullet) +\delta_{\tu{offset}})$ is just the pullback of $h$ 
with a reflection at the boundary of disk $D_2$ plus  
complex conjugation and a final translation. Thus from the fact that $h$ is 
holomorphic on $D_5$ follows that 
$h(R_2(\bullet) +\delta_{\tu{offset}})$ is holomorphic on $\C \setminus D_{2}$.
For the same reason the term $h(R_3(\bullet) +\delta_{\tu{offset}})$ is 
holomorphic on $\C\setminus D_3$. Let us next consider the term 
$\left[-R_2'(u)\right]^{s}$. From the form (\ref{eq:R_theta_phi_as_reflection}) 
one deduces directly that $-R_2'(\bullet)$ is a nonzero holomorphic 
function on $\C \setminus\{m_2\}$, consequently 
$\left[-R_2'(u)\right]^{s}$ can be extended to every split plane 
$\C\setminus l$ where $l$ is a half line starting at the center $m_2$ and 
going to infinity. Analogously $[-R'_3]^s$ can be extended to every split
plane without a line starting at $m_3$. We can therefore 
extend $\mathcal L_{s,z;\ell}^{(1)} h$ to any disk centered around $m_1$ 
that does not intersect $D_2$ nor $D_3$ and in particular to the disk $E_1$ 
as defined above. Analogously any of the other 
functions $\mathcal L_{s,z;\ell}^{(j)} h$ can be extended from $D_j$ to $E_j$. 
This extension will now allow us 
to construct a nuclear representation of the operators $\mathcal L^{(j)}_{s,z;\ell}$ 
and control the appearing terms.

Let us denote by $C_j$ the circle of radius $1$ around $m_j$. As $C_j$ is 
strictly contained in $E_j$ we can write with the holomorphic extension 
to $E_j$ and Cauchy's integral formula for any $\ell>\ell_0$ and any 
$u \in D_j$
\[
 (\mathcal L^{(j)}_{s,z;\ell} h)(u)=\frac{1}{2\pi i}\int\limits_{C_j}\frac{(\mathcal L^{(j)}_{s,z;\ell} h)(\xi)}{\xi-u} d\xi.
\]
As we know for any $\xi \in C_j$ and $u\in D_j$ that $|u-m_j|<|\xi-m_j|$ 
we can use the geometric series to write
\begin{eqnarray*}
 (\mathcal L^{(j)}_{s,z;\ell} h)(u)&=& \frac{1}{2\pi i}\int\limits_{C_j}\frac{(\mathcal L^{(j)}_{s,z;\ell} h)(\xi)}{\xi-m_j}\left(1-\frac{u-m_j}{\xi-m_j} \right)^{-1} d\xi\\
 &=& \sum\limits_{n=0}^\infty \frac{1}{2\pi i}\int\limits_{C_j}\frac{(\mathcal L^{(j)}_{s,z;\ell} h)(\xi)}{\xi-m_j}\left(\frac{u-m_j}{\xi-m_j} \right)^{n} d\xi\\
 &=& \sum\limits_{n=0}^\infty \tilde \alpha^{(j)}_n(h) \tilde v_n^{(j)}(u),
\end{eqnarray*}
where 
\begin{equation}\label{eq:alpha_v_tilde}
 \tilde \alpha_n^{(j)}(h):= \frac{1}{2\pi i}\int\limits_{C_j}\frac{(\mathcal L^{(j)}_{s,z;\ell} h)(\xi)}{(\xi-m_j)^{n+1}}d\xi ~\tu{ and }~ \tilde v_n^{(j)}(u) =(u-m_j)^n.
\end{equation}
We can finally normalize the elements $\tilde v_n^{(j)}$ with respect 
to the supremum norm and $\tilde \alpha_n^{(j)}$ with respect to the operator norm
as a linear operator $\mathcal A(D)\to\C$ and we obtain the nuclear representation 
\[
 (\mathcal L^{(j)}_{s,z;\ell} h)(u) = \sum\limits_{n=0}^\infty \lambda_n^{(j)}\alpha^{(j)}_n(h) v_n^{(j)}(u)
\]
with $\lambda_n^{(j)}= \|\tilde \alpha_n^{(j)}\|\|\tilde v_n^{(j)}\|$.

Equation (\ref{eq:alpha_v_tilde}) also allows us to obtain estimates on 
$\lambda_n^{(j)}$. Recall that $r_j$ was the radius of disk $D_j$ so we have
\[
 \|\tilde v_n^{(j)}\|_{\infty,D_j} =r_j^n.
\]
In order to bound $\|\tilde \alpha_n^{(j)}\|$ first calculate for any 
$h\in \mathcal A_\infty$ that 
\[
 |\tilde \alpha_n^{(j)}(h)| \leq \frac{1}{2\pi} \|\mathcal L_{s,z;\ell}^{(j)}h\|_{\infty,E_j} \leq \frac{1}{2\pi} \max\limits_{i:\,j\rightsquigarrow i}\|V(\phi_{j,i}(\bullet);s,z)\|_{\infty, E_j} \|h\|_{\infty,D}
\]
so putting the two bounds together we get
\begin{eqnarray}
  \lambda_n^{(j)} &=& \|\alpha_n^{(j)}\|\|v_n^{(j)}\| \nonumber \\  
                  &\leq& r_j^{n} \frac{1}{2\pi }  
                  \max\limits_{i:\,j\rightsquigarrow i}
                  \|V(\phi_{j,i}(\bullet);s,z)\|_{\infty, E_j} 
                  \label{eq:lambda_j_bound}.
\end{eqnarray}
We have thus derived the nuclear representation of $\mathcal L^{(j)}_{s,z;\ell}$ 
and also obtained explicit bounds on the $\lambda_n^{(j)}$. In order to control 
the nuclear representation of the full operator $\mathcal L_{s,z;\ell}$ we have to sum up the six 
operators $\mathcal L_{s,z;\ell}^{(j)}$ with $1\leq j\leq 6$. We arrange the 
different summands such that 
\[
 \mathcal L_{s,z;\ell} h = \sum\limits_{n=0}^\infty \lambda_n \alpha_n(h)v_n
\]
with 
\[
 \lambda_{6n+j}=\lambda_n^{(j)},~~\alpha_{6n+j}=\alpha_n^{(j)} \tu{ and } v_{6n+j}=v_n^{(j)}.
\]
If we define 
\[
r(\ell):=\max_j (r_j)~\tu{ and }~K(s,z;\ell):=\frac{1}{2\pi}\max_{j\rightsquigarrow i}\|V(\phi_{j,i}(\bullet);s,z)\|_{\infty, E_j} 
\]
 then we have the explicit bound
\begin{equation}\label{eq:lambda_bound}
 \lambda_n\leq K(s,z;\ell)r(\ell)^{\lfloor n/6 \rfloor}.
\end{equation}
We can now use the Grothendieck formula (\ref{eq:Grothendieck_taylor_coeff}) 
as well as the Hadamard bound on the $k\times k$ matrices with entries lower 
or equal to one and obtain
\begin{eqnarray*}
 |\tilde d_{\mathbf n}^{(k)}(s,z,\ell)| &=& \sum_{m_1<\dots<m_k} \lambda_{m_1}\ldots\lambda_{m_k} \det\left[(\alpha_{m_p}(v_{m_q}))_{p,q=1}^k\right]\\
 &\leq& k^{k/2}  K(s,z;\ell)^{k}\sum_{m_1<\dots<m_k} r(\ell)^{\lfloor m_1/6 \rfloor+\ldots+ \lfloor m_k/6\rfloor}
\end{eqnarray*}
which finishes the proof of Lemma~\ref{lem:zeta_trunctaion}.
\end{proof}
For the rest of the proof of Theorem~\ref{thm:rescaled_zeta_limit} there remain two steps
to be done: First we will show that the bounds in Lemma~\ref{lem:zeta_trunctaion} allow 
to uniformly truncate the cycle expansion after the sixth order and secondly we will show
that the finitely many remaining terms converge against the polynomial. For both steps, the 
following lemma will be useful. 
\begin{lem}\label{lem:limit_of_potential_1}
 For any $j\rightsquigarrow i$ and any bounded domain $B \subset \C^2$ we have
 \[
  \lim\limits_{\ell\to\infty} V(\phi_{j,i}(u),s/\ell,\sqrt{z}) 
  = z^{n_{j,i}/2} e^{-\kappa_{i}s}
 \]
 uniformly for $u\in E_i$ and $(s,z)\in B$.
\end{lem}
\begin{proof}
 Recall that
\[
 V(\phi_{j,i}(u);s/\ell,\sqrt z) =
 z^{n_{j,i}/2}\left[-(\phi_{j,i}^{-1})'(\phi_{j,i}(u))\right]^{s/\ell} = 
 z^{n_{j,i}/2}\left[-\phi_{j,i}'(u)\right]^{s/\ell}
\]
and calculate that for $1\leq i\leq 3$
\[
 -\phi_{j,i}'(u) = \frac{r_i}{(u-\delta_{\tu{offset}}-m_i)^2}
\]
while for $4\leq i\leq 6$
\[
 -\phi_{j,i}'(u) = \frac{r_{i-3}}{(u-m_{i-3})^2}.
\]
By the definition of the family of flow-adapted IFS in 
Lemma~\ref{lem:family_flow_IFS} and the extended disks $E_i$ we know
that  there exist two constants
$0<c<C$ such that for any $j\rightsquigarrow i$, any  $u\in E_j$ we 
have $c<|u-\delta_{\tu{offset}}-m_i|<C$ if $i\leq 3$ and 
$c<|u-m_{i-3}|<C$ if $i> 3$. Furthermore the asymptotics (\ref{eq:asymptotic_r_j})
for $r_i$ gives us another pair of constants $0<c<C$ 
with $c<r_i e^{\kappa_i \ell}<C$. Together with the fact that $s$ can vary only in 
a bounded set $\tu{Pr}_s B$, which is the projection of $B$ to the $s$ variable, 
this gives the existence of constants $0<c<C$ with
\begin{equation}\label{eq:cCestimate}
 c<\left(-\phi'_{j,i}(u)e^{\kappa_i \ell}\right)^s<C ~~\forall~ u\in E_j, \ell>\ell_0 \tu{ and } s\in \tu{Pr}_s B.
\end{equation}
In order to use this inequality we calculate
\begin{eqnarray*}
 \left\|V(\phi_{j,i}(u),s/\ell,\sqrt z) 
  - z^{n_{j,i}/2} e^{-\kappa_{i}s}\right\|_{\infty, E_j\times B} &=&\left\|z^{n_{j,i}/2} \left(\left(-\phi_{j,i}'(u)\right)^{s/\ell}
  - e^{-\kappa_{i}s}\right)\right\|_{\infty, E_j\times B}\\
  &\leq&\left\|z^{n_{j,i}/2} e^{-\kappa_{i}s}\right\|_{\infty, B} \left\|\left(\left(-\phi_{j,i}'(u)e^{\kappa_i \ell}\right)^{s/\ell}
 - 1\right)\right\|_{\infty, E_j\times B}.\\
\end{eqnarray*}
While the first term $\left\|z^{n_{j,i}/2} e^{-\kappa_{i}s}\right\|_{\infty, B}$ 
is bounded and independent of $\ell$, the uniform bounds (\ref{eq:cCestimate}) imply
that the second term converges to zero.
\end{proof}

\begin{lem}\label{lem:convergence_of_leading_cycle}
Let $B\subset \C^2$ be any bounded domain, and
\[
 d_{\mathbf{n}}(s,z;\ell)=1+\sum\limits_{k=1}^{\infty} \tilde d_{\mathbf n}^{(k)}(s,\sqrt{z};\ell).
\]
the series expansion (\ref{eq:gen_zeta:limit_cycle}) of the generalized zeta function
from Lemma~\ref{lem:zeta_trunctaion}. Then
\begin{equation}\label{eq:uniform_reminder_convergence}
 \lim_{\ell \to\infty}\left\|\sum\limits_{k>6} \tilde d_{\mathbf n}^{(k)}(s/\ell,\sqrt{z};\ell)\right\|_{\infty,B}=0,
\end{equation}
i.e.\,the rescaled generalized zeta function $d_{\mathbf{n}}(s/\ell,z;\ell)$ 
converges uniformly to the finitely truncated sum 
$1+\sum\limits_{k=1}^{6} \tilde d_{\mathbf n}^{(k)}(s/\ell,\sqrt{z};\ell)$.
\end{lem}
\begin{proof}
In a first step we show that Lemma~\ref{lem:limit_of_potential_1} implies
a bound for $\|K(s/\ell,\sqrt {z},\ell)\|_{\infty,B}$. Recall from the 
definition (\ref{eq:gamma_l_K_szl}) that 
\[
\|K(s/\ell,\sqrt {z},\ell)\|_{\infty,B} = 
\max\limits_{j\rightsquigarrow i} \left\|V(\phi_{j,i}(u);s/\ell,\sqrt{z}) \right\|_{\infty,B\times E_j}.
\]
Now Lemma~\ref{lem:limit_of_potential_1} implies that 
\[
\left\|V(\phi_{j,i}(u);s/\ell,\sqrt{z}) -z^{n_{j,i}/2}e^{-\kappa_i s}\right\|_{\infty,B\times E_j}<C_{j,i}
\]
for all $\ell>\ell_0$ and thus
\[
\|K(s/\ell,\sqrt {z},\ell)\|_{\infty,B} \leq 
\max\limits_{j\rightsquigarrow i} \left(\left\|z^{n_{j,i}/2}e^{-\kappa_i s}\right\|_{\infty,B} +C_{j,i}\right)=:K_B.
\]

As a second step we use that $\lfloor n/6\rfloor \geq (n-5)/6$ and obtain
\begin{eqnarray*}
 \|\tilde d_{\mathbf n}^{(k)}(s/\ell,z,\ell)\|_{\infty,B} &\leq& k^{k/2}  K_B^{k} r(\ell)^{-5k/6} \sum_{m_1<\dots<m_k} \left(r(\ell)^{1/6}\right)^{m_1 +\ldots+ m_k}.\\
\end{eqnarray*}
Setting $\tilde r(\ell)=r(\ell)^{1/6}$ and using the Euler formula this gives
\[
 \|\tilde d_{\mathbf n}^{(k)}(s/\ell,z,\ell)\|_{\infty,B} \leq k^{k/2}  K_B^{k} \tilde r(\ell)^{-5k}
 \frac{\tilde r(\ell)^{k(k-1)/2}}{(1-\tilde r(\ell))\ldots(1-\tilde r(\ell)^k)}
\]
which allows us to obtain an estimate for $k\geq12$
\[
 \left\|\sum\limits_{k=12}^\infty \tilde d_{\mathbf n}^{(k)}(s/\ell,z,\ell)\right\|_{\infty,B} \leq \tilde r(\ell)
 \sum\limits_{k=12}^\infty k^{k/2}  K_B^{k} \
 \frac{\tilde r(\ell)^{k\left(\frac{k-1}{2} -5 -\frac{1}{k}\right)}}{(1-\tilde r(\ell))\ldots(1-\tilde r(\ell)^k)}.
\]
If $k\geq 12$ then we have $\left(\frac{k-1}{2} -5 -\frac{1}{k}\right)>0$ 
and every term in the sum is uniformly bounded for all $\ell\geq\ell_0$. Furthermore 
since the terms in the series decay super-exponentially in $k$ thanks 
to the term $\tilde r(\ell)^{k^2}$ the series converges with a uniform bound
and the factor $\tilde r(\ell)$ in front assures the convergence to zero. 

It remains thus to prove that the coefficients for $7\leq k\leq 11$ vanish. This can 
be seen as follows. Note that we can estimate
\[
\sum_{m_1<\dots<m_k} r(\ell)^{\lfloor m_1/6 \rfloor+\ldots+ \lfloor m_k/6\rfloor} 
\leq \left(\sum_{m>0} r(\ell)^{\lfloor m/6 \rfloor}\right)^k\leq  C_k
\]
for all $\ell>\ell_0$ with $C_k:=\max\limits_{\ell>\ell_0} (6/(1-r(\ell))^k)$ 
independent of $\ell$. We can thus write for $k>6$
\begin{eqnarray*}
 \sum_{m_1<\dots<m_k} r(\ell)^{\lfloor m_1/6 \rfloor+\ldots+ \lfloor m_k/6\rfloor} &=& 
 \sum_{m_1<\dots<m_{k-1}} r(\ell)^{\lfloor m_1/6 \rfloor+\ldots+ 
 \lfloor m_{k-1}/6\rfloor} \sum\limits_{m_k>m_{k-1}} r(\ell)^{\lfloor m_k/6\rfloor}\\
 &\leq&C_{k-1}\sum\limits_{m_k\geq 6} r(\ell)^{\lfloor m_k/6\rfloor}\\
 &=&C_{k-1}\frac{6\tilde r(\ell)}{1-\tilde r(\ell)}.
\end{eqnarray*}
Here we used crucially that from $k\geq 7$ we have $m_k\geq 6$ and thus can obtain the 
bound on the sum over $m_k$. We have finally shown (\ref{eq:uniform_reminder_convergence}) and
finished the proof of Lemma~\ref{lem:convergence_of_leading_cycle}.
\end{proof}
In order to prove Theorem~\ref{thm:rescaled_zeta_limit} it finally remains to prove that
\[
 \lim\limits_{\ell\to\infty}\left\|1+\sum\limits_{k=1}^6 \tilde d^{(k)}_{\mathbf n}(s/\ell,\sqrt{z};\ell) - P_{n_1,n_2,n_3}(ze^{-s}) \right\|_{\infty,B} =0.
\]
Recall that from (\ref{eq:cycle_coeff_formula}) we have
\begin{equation}\label{eq:cycle_coeff_d_tilde}
 \tilde d^{(k)}_{\mathbf n}(s/\ell,\sqrt{z};\ell) = \sum\limits_{m=1}^{k}\left(\sum\limits_{
                           (n_1,\dots,n_m)\in P(k,m)}
                           \frac{(-1)^m}{m!}\prod\limits_{l=1}^{m} \frac{1}{n_l} \sum\limits_{w\in \mathcal W^{cl}_{n_l}} \frac{V_w(u_w;s/\ell,\sqrt z)}{1-\phi_w'(u_w)}\right),
\end{equation}

so we can explicitly calculate the cycle expansion coefficients in terms of 
dynamical quantities of the holomorphic IFS. As the symbolic dynamics of the
flow-adapted IFS directly implies that the set of closed words is empty for
uneven word length this drastically reduces the complexity of the calculations:
First of all only coefficients $\tilde d_{\mathbf n}^{(k)}$ with $k=2,4,6$ can 
be nonzero, because otherwise at least one summand $n_l$ is uneven and consequently
one factor in the product $\prod_{l=1}^{m}$ is zero. Additionally this condition
reduces the number or possible tuples $(n_1,\ldots,n_m)$ which lead to nonzero
contributions drastically: For $k=2$ it remains only the one tuple $(2)$, for $k=4$
there are two possibilities $(4)$ and $(2,2)$ and for $k=6$ there are four 
possible tuples, namely $(6)$, $(4,2)$, $(2,4)$ and $(2,2,2)$. Even if each coefficient
is only given by an explicit finite sum and even if the complexity of this sums is 
tremendously reduced by the above discussion it remains still very 
complex as the number of closed words 
increases exponentially. As $\# \mathcal W_2^{cl}=12$, $\# \mathcal W_4^{cl}=36$ and
$\# \mathcal W_6^{cl}=132$, the coefficient $d_{\mathbf n}^{(6)}$ would a priori be 
given by $132+ 12\cdot	36+36\cdot 12+12^3=2724$ summands. The following Lemma however allows to 
reduce the complexity in the limit $\ell\to\infty$ strongly.
\begin{lem}\label{lem:limit_of_potential}
 Let us define for any $n\in \N$
 \begin{equation}\label{eq:order_function_for_words}
  \mathbf n:\Abb{\mathcal W_n^{cl}}{\N}{(w_0,\ldots,w_n)}{
  \frac{1}{2}\sum_{r=0}^{n-1}n_{w_r,w_{r+1}}}.
 \end{equation}
Then for any finite closed word $w\in\mathcal W_n^{cl}$ for the symbolic dynamics
 of the flow-adapted IFS and any $B\subset \C^2$ we have
 \begin{equation}\label{eq:limit_of_w_dependent_terms}
  \lim\limits_{\ell\to\infty} \left\|\frac{V_w(u_w;s/\ell,\sqrt z)}{1-\phi_w'(u_w)}- \left(z e^{-s}\right)^{\mathbf n(w)}\right\|_{\infty,B}=0.
 \end{equation}
\end{lem}
\begin{rem}
 The notation of the application $\mathbf n: \mathcal W_k^{cl}\to \N$ does not
 only coincide by chance with the notation of the order function on the closed geodesic
 which was defined in (\ref{eq:order_function_on_geodesics}). In fact
 if $w$ is a prime word, then Proposition~\ref{prop:SymIFS_orbit_geodesic_equiv:limit}
 associates this word to a primitive geodesic in 
 $\mathcal P_{X_{n_1,n_2,n_3}}(\ell)$. The definition (\ref{eq:order_function_for_words})
 of the order function restricted to the subset of prime words 
 $\mathcal W_k^{\tu{prime}}$ is then equal to the order function 
 (\ref{eq:order_function_on_geodesics}) on primitive geodesics
 with respect to this identification. 
\end{rem}

\begin{proof}
Let us first note that from the fact that $r_i\to 0$ we conclude for 
any $w\in \mathcal W_n^{cl}$ that $\lim\limits_{\ell\to\infty}\phi'_w(u_w)=0$.
It thus only remains to handle the term $V_w(u_w;s/\ell,\sqrt z)$ and as a
first step we note that as $u_w\in D_{w_0}\subset E_{w_0}$, 
Lemma~\ref{lem:limit_of_potential_1} implies that
 \begin{equation}\label{eq:limit_of_potential}
  \lim\limits_{\ell\to\infty} \left\| V(\phi_{w_0,w_1}(u_w),s/\ell,\sqrt{z}) 
  - z^{n_{w_0,w_1}/2} e^{-\kappa_{w_1}s}\right\|_{\infty, B} =0.
 \end{equation}
From the definition (\ref{eq:iterated_product}) of the iterated product we obtain
\begin{eqnarray*}
 V_w(u_w;s/\ell ,\sqrt z)&:=&\prod\limits_{k=1}^n V(\phi_{w_{0,k}}(u_w);s/\ell ,\sqrt z)\\
 &=&\prod\limits_{k=1}^n V(\phi_{w_{k-1},w_k}(u_{\sigma_L^{(k-1)} w});s/\ell ,\sqrt z).
\end{eqnarray*}
Here we used that the dynamics on the fixed points is conjugated to the shift 
operation (see (\ref{eq:conjugate_to_shift})). Plugging in (\ref{eq:limit_of_potential})
we obtain
\[
\lim_{\ell\to\infty}\left\|V_w(u_w;s/\ell, \sqrt z)) - z^{\frac{1}{2}\left(\sum_{k=1}^n n_{w_{k-1},w_k}\right)} e^{-s\left(\sum_{k=1}^n \kappa_{w_k}\right)}\right\|_{\infty, B} =0.
\]
Thus it only remains to show that $\sum_{k=1}^n \kappa_{w_k} = \mathbf n(w)$ 
in order to finish the proof. This can finally be seen as follows. First one 
checks  that for any $i\rightsquigarrow j$ we have $\kappa_i+\kappa_j=n_{i,j}$.
Secondly as $w_0=w_n$ we can write 
\[
 \sum\limits_{k=1}^n \kappa_{w_k}=\frac{1}{2}\sum\limits_{k=1}^n \kappa_{w_{k-1}}+\kappa_{w_k} =
 \frac{1}{2}\sum\limits_{k=1}^n n_{w_{k-1},w_k} = \mathbf n(w).
\]
\end{proof}
\begin{rem}
 The identity  $\sum_{k=1}^n \kappa_{w_k} = \mathbf n(w)$ shows that we could also
 have taken another definition of the potential functions for the generalized zeta
 functions namely $V(u)=z^{2\kappa_j} [-(\phi^{-1})'(u)]^{-s}$ for 
 $u\in\phi_{i,j}(D_i)$. Note that however
 the $\kappa_i$ are only positive if the $n_i$ fulfill the triangle condition as 
 defined in Lemma~\ref{lem:family_flow_IFS}.
 In those cases where it is not satisfied the analyticity of the generalized 
 zeta function would have been much harder to proof, so we chose the definition 
 by the $n_{i,j}$.
\end{rem}

We are now ready to proof Theorem~\ref{thm:rescaled_zeta_limit}.
\begin{proof}[Proof of Theorem~\ref{thm:rescaled_zeta_limit}]
Lemma \ref{lem:limit_of_potential} implies that the $w$ dependent terms in 
(\ref{eq:cycle_coeff_d_tilde}) depend only on $\mathbf n(w)$ 
in the limit $\ell\to\infty$.
We thus introduce for any $k,n\in N$ the sets 
\[
 \mathcal W_k^{cl}(n):= \left\{ w\in\mathcal W_k^{cl},~\mathbf n(w)=n\right\}.
\]
and observe that the relevant set of words split into
\begin{eqnarray}
 \mathcal W_2^{cl} &=& \mathcal W_2^{cl}(n_1)\cup \mathcal W_2^{cl}(n_2)\cup \mathcal W_2^{cl}(n_3)
 \label{eq:word_split2} \\
 \mathcal W_4^{cl} &=& \mathcal W_4^{cl}(2n_1)\cup \mathcal W_4^{cl}(2n_2)\cup \mathcal W_4^{cl}(2n_3)\cup \nonumber\\
                      &&\mathcal W_4^{cl}(n_1+n_2)\cup \mathcal W_4^{cl}(n_2+n_3)\cup \mathcal W_4^{cl}(n_1+n_3) \label{eq:word_split4}\\
 \mathcal W_6^{cl} &=& \mathcal W_6^{cl}(3n_1)\cup \mathcal W_6^{cl}(3n_2)\cup \mathcal W_6^{cl}(3n_3)
                        \cup \nonumber\\
                        &&\mathcal W_6^{cl}(2n_1+n_2)\cup \mathcal W_6^{cl}(n_1+2n_2)\cup 
                        \mathcal W_6^{cl}(2n_2+n_3)\cup \mathcal W_6^{cl}(n_2+2n_3)\cup \nonumber\\
                        &&\mathcal W_6^{cl}(2n_1+n_3)\cup \mathcal W_6^{cl}(n_1+2n_3)\cup \mathcal W_6^{cl}(n_1+n_2+n_3)\label{eq:word_split6}
\end{eqnarray}
where the number of elements per set is given by
\begin{equation}\label{eq:word_number}
\begin{array}{rcll}
\#\mathcal W_2(n_j) =\#\mathcal W_4(2n_j)=\#\mathcal W_6(3n_j) &=&4&\forall 1\leq j\leq 3\\
\#\mathcal W_4(n_i+n_j)&=& 8 &\forall 1\leq i,j\leq 3\tu{ with }i\neq j\\
\#\mathcal W_6(2n_i+n_j)&=&12 &\forall 1\leq i,j\leq 3\tu{ with }i\neq j\\
\#\mathcal W_6(n_1+n_2+n_3)&=&48 &
\end{array}
\end{equation}
One can convince oneself from the validity of these formulas by geometric arguments.
For example the only closed geodesics, that intersect only two of the 
blue lines in Figure~\ref{fig:shadow_geodesic} are those who make one circle 
around one of the  three funnels. As the closed words 
correspond to closed geodesics, the closed words
of order two split according to (\ref{eq:word_split2}). Around each funnel there 
are two different geodesics (one in each sense of orientation) and each geodesic 
is encoded by two different words which leads to $\mathcal W_2^{cl}(n_i) =4$. 
All other results can be understood by similar arguments, the easiest way to 
calculate (\ref{eq:word_split2})-(\ref{eq:word_number}) is however to 
solve the finite combinatorial problem exactly with a computer.

With this data it is a straight forward task to calculate that
\begin{eqnarray}
 \left\| \tilde d_{\mathbf n}^{(2)}(s/\ell,\sqrt z;\ell) +2\left[ (ze^{-s})^{n_1}+(ze^{-s})^{n_2}+(ze^{-s})^{n_3} \right]\right\|_{\infty,B}&=&0~~\label{eq:cycle_coeff_limit_2}\\
 \bigg\| \tilde d_{\mathbf n}^{(4)}(s/\ell,\sqrt z;\ell) - \Big[
 (ze^{-s})^{2n_1}+(ze^{-s})^{2n_2}+(ze^{-s})^{2n_3}+&&\nonumber\\
2\left((ze^{-s})^{n_1+n_2}+(ze^{-s})^{n_1+n_3}+(ze^{-s})^{n_2+n_3}
 \right)\Big] \bigg\|_{\infty,B}&=&0~~\label{eq:cycle_coeff_limit_4}\\
 \left\| \tilde d_{\mathbf n}^{(6)}(s/\ell,\sqrt z;\ell) + 4(ze^{-s})^{n_1+n_2+n_3} \right\|_{\infty,B}&=&0~~\label{eq:cycle_coeff_limit_6}
\end{eqnarray}
Equation (\ref{eq:cycle_coeff_limit_2}) is seen immediately because as discussed
above the only possible tuple $(n_1,\ldots,n_m)$ is the one-tuple $(2)$. The 
next equation (\ref{eq:cycle_coeff_limit_4}) can be seen as follows: First we
split (\ref{eq:cycle_coeff_d_tilde}) according to the two possible tuples 
$(4)$ and $(2,2)$
\[
\tilde d_{\mathbf n}^{(4)}(s/\ell,\sqrt z;\ell) = \underbrace{-\frac{1}{4}\left(\sum\limits_{w\in\mathcal W_4^{cl}} 
\frac{V_w(u_w;s/\ell,\sqrt z)}{1-\phi_w'(u_w)}\right)}_{(A)}+ \underbrace{\frac{1}{2!}\left(\frac{1}{2}\sum\limits_{w\in\mathcal W_2^{cl}}\frac{V_w(u_w;s/\ell,\sqrt z)}{1-\phi_w'(u_w)}\right)^2}_{(B)}.
\]
Next we treat the parts (A) and (B) separately. For (A) we use 
(\ref{eq:word_split4}) and obtain
\begin{eqnarray*}
 (A)&=&-\frac{1}{4}\bigg(\sum\limits_{w\in\mathcal W_4^{cl}(2n_1)} \left[\frac{V_w(u_w;s/\ell,\sqrt z)}{1-\phi_w'(u_w)}-(ze^{-s})^{2n_1}\right]+\\
 &&\sum\limits_{w\in\mathcal W_4^{cl}(2n_2)} \left[\frac{V_w(u_w;s/\ell,\sqrt z)}{1-\phi_w'(u_w)}-(ze^{-s})^{2n_2}\right]+\\
 &&\sum\limits_{w\in\mathcal W_4^{cl}(2n_3)} \left[\frac{V_w(u_w;s/\ell,\sqrt z)}{1-\phi_w'(u_w)}-(ze^{-s})^{2n_3}\right]+\\
 &&\sum\limits_{w\in\mathcal W_4^{cl}(n_1+n_2)} \left[\frac{V_w(u_w;s/\ell,\sqrt z)}{1-\phi_w'(u_w)}-(ze^{-s})^{n_1+n_2}\right]+\\
 &&\sum\limits_{w\in\mathcal W_4^{cl}(n_2+n_3)} \left[\frac{V_w(u_w;s/\ell,\sqrt z)}{1-\phi_w'(u_w)}-(ze^{-s})^{n_2+n_3}\right]+\\
 &&\sum\limits_{w\in\mathcal W_4^{cl}(n_1+n_3)} \left[\frac{V_w(u_w;s/\ell,\sqrt z)}{1-\phi_w'(u_w)}-(ze^{-s})^{n_1+n_3}\right]-\\
 && \Big[4(ze^{-s})^{2n_1}+4(ze^{-s})^{2n_2}+4(ze^{-s})^{2n_3}+ 8(ze^{-s})^{n_1+n_2}+ 8(ze^{-s})^{n_2+n_3}+ 8(ze^{-s})^{n_1+n_3}\Big]\bigg).
\end{eqnarray*}
In order to treat (B) we use (\ref{eq:word_split2}) and calculate
\begin{eqnarray*}
 (B)&=&+\frac{1}{2}\bigg(\frac{1}{2}\sum\limits_{w\in\mathcal W_2^{cl}(n_1)} \left[\frac{V_w(u_w;s/\ell,\sqrt z)}{1-\phi_w'(u_w)}-(ze^{-s})^{n_1}\right]+\\
 &&\frac{1}{2}\sum\limits_{w\in\mathcal W_2^{cl}(n_2)} \left[\frac{V_w(u_w;s/\ell,\sqrt z)}{1-\phi_w'(u_w)}-(ze^{-s})^{n_2}\right]+\\
 &&\frac{1}{2}\sum\limits_{w\in\mathcal W_2^{cl}(n_3)} \left[\frac{V_w(u_w;s/\ell,\sqrt z)}{1-\phi_w'(u_w)}-(ze^{-s})^{n_3}\right]+\\
 && \Big[2(ze^{-s})^{n_1}+2(ze^{-s})^{n_2}+2(ze^{-s})^{n_3}\Big]\bigg)^2.
\end{eqnarray*}
Note that in both equations of (A) and (B), respectively, all terms except the 
last line converge uniformly to $0$ on the set $B\subset \C^2$. So the limit
$\ell\to\infty$ the coefficient $\tilde d_{\mathbf n}^{(4)}(s/\ell,\sqrt z;\ell) $
converges uniformly to 
\begin{eqnarray*}
 &&-\frac{1}{4} \left[4(ze^{-s})^{2n_1}+4(ze^{-s})^{2n_2}+4(ze^{-s})^{2n_3}+ 8(ze^{-s})^{n_1+n_2}+ 8(ze^{-s})^{n_2+n_3}+ 8(ze^{-s})^{n_1+n_3}\right]\\
 &&+\frac{1}{2} \left[2(ze^{-s})^{n_1}+2(ze^{-s})^{n_2}+2(ze^{-s})^{n_3}\right]^2\\
 &&=(ze^{-s})^{2n_1}+(ze^{-s})^{2n_2}+(ze^{-s})^{2n_3}+2(ze^{-s})^{n_1+n_2}+ 2(ze^{-s})^{n_2+n_3}+ 2(ze^{-s})^{n_1+n_3}
\end{eqnarray*}
which proves (\ref{eq:cycle_coeff_limit_4}). By a completely analogous 
but more tedious calculation we can show (\ref{eq:cycle_coeff_limit_6}).

Finally we can put (\ref{eq:cycle_coeff_limit_2}), (\ref{eq:cycle_coeff_limit_4}) 
and (\ref{eq:cycle_coeff_limit_6}) together and obtain 
(\ref{eq:rescaled_zeta_polynomial_convergence}) which finishes the proof of 
Theorem~\ref{thm:rescaled_zeta_limit}.
\end{proof}
\begin{rem}
\begin{figure}
\centering
        \includegraphics[width=1.1\textwidth]{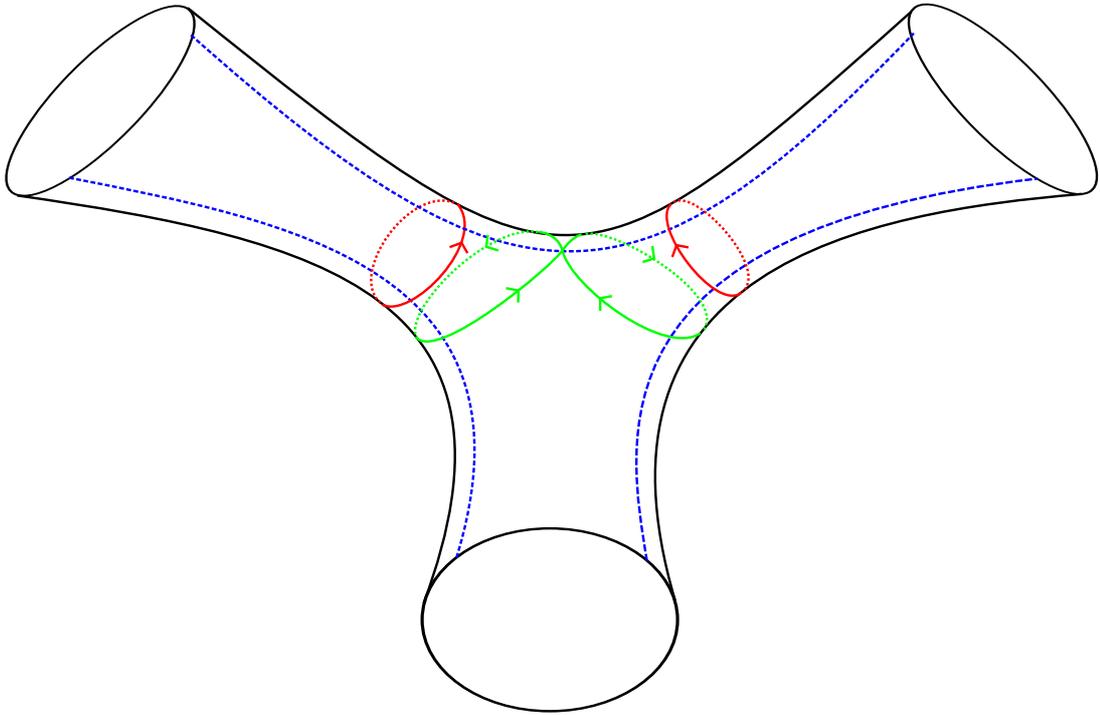}
\caption{Schematic sketch of a 3-funneled Schottky 
surface. The blue, dashed lines indicate the cut lines of the Poincaré section which 
would correspond to the flow-adapted IFS. In red we see two geodesics which 
make one turn around one funnel each. They correspond to the closed words of 
length 2 $w^{(a)}=(1,5,1)$ and $w^{(a)}=(3,5,3)$. In green we see a geodesic 
which winds around both funnels in an eight-like shape. It corresponds to the 
closed word $w^{(c)}=(1,5,3,5,1)$ of length 4. Note that 
$\mathbf n(w^{(a)})+\mathbf n(w^{(b)}) = \mathbf n(w^{(c)})$.
}
\label{fig:shadow_geodesic}
\end{figure}
 Note that the limit form of $\frac{V_w(u_w;s/\ell,\sqrt z)}{1-\phi_w'(u_w)}$ 
 in (\ref{eq:limit_of_w_dependent_terms}) of 
 Lemma~\ref{lem:limit_of_potential} not only allows to group many terms together
 but also allows to take advantage of a systematic canceling. For example in 
 the calculation of $\tilde d_{\mathbf n}^{(4)}(s/\ell,\sqrt z;\ell)$ the terms
 $(ze^{-s})^{(n_1+n_2)}$ appears as limits of two different geodesics. First they
 appear in the term (A) as limits of the eight-shaped geodesics which turn 
 around the funnels of width $n_1$ and $n_2$ (see green geodesic in 
 Figure~\ref{fig:shadow_geodesic}). Secondly they appear
 in (B) as the product of the geodesic which turns once around the funnel of 
 width $n_1$ with another geodesic which turns once around the funnel of width 
 $n_2$ (see the two red geodesics in Figure~\ref{fig:shadow_geodesic}). 
 As both terms appear with different signs they cancel 
 each other to a  big extend. Note that this cancellation is not exactly 
 true for finite $\ell$. In the setting of the physical quantum 3-disk system
 it has however been argued that this cancellation is
 approximately true. The mechanism that the contribution of longer orbits is
 approximately canceled by a combination of shorter orbits which Cvitanovic and 
 Eckhardt call \emph{shadowing orbits} has been identified in physics literature
 as the key mechanism for the fast convergence of the cycle expansion. 
 Lemma~\ref{lem:limit_of_potential} can thus also be seen as a proof that
in the limit $\ell\to\infty$ this approximation becomes exact on Schottky 
surfaces.
\end{rem}

Theorem~\ref{thm:location_res} on the location of the rescaled resonances 
now follows directly from Theorem~\ref{thm:rescaled_zeta_limit}.
\begin{proof}[Proof of Theorem~\ref{thm:location_res}]
Recall that 
\[
 \mathcal N_{n_1,n_2,n_3}=\{s\in\C,~P_{n_1,n_2,n_3}(e^{-s})=0\}
\]
and
\[
 \widetilde{\tu{Res}}_{n_1,n_2,n_3}(\ell):=\{s\in\C,~s/\ell \in 
 \tu{Res}(X_{n_1,n_2,n_3}(\ell))\}. 
\]
If $U\subset \C$ is a domain whose boundary $\partial U$ is disjoint with
$\mathcal{N}_{n_1,n_2,n_3}$ then the argument principle implies that
\[
\#\Big( U\cap\mathcal N_{n_1,n_2,n_3}\Big) = \frac{1}{2\pi i}
    \int\limits_{\partial U} \frac{f'(s)}{f(s)}ds
\]
with $f(s)=P_{n_1,n_2,n_3}(e^{-s})$. On the other hand 
\[
 \#\Big( U\cap\widetilde{\tu{Res}}_{n_1,n_2,n_3}(\ell)\Big) = 
 \int\limits_{\partial U} \frac{\frac{d}{ds}Z_{X_{n_1,n_2,n_3}(\ell)}(s/\ell)}{Z_{X_{n_1,n_2,n_3}(\ell)}(s/\ell)}ds
\]
and Theorem~\ref{thm:rescaled_zeta_limit} implies that
\[
Z_{X_{n_1,n_2,n_3}(\ell)}(s/\ell) = d_{\mathbf n}(s/\ell,1)\to f(s) 
\]
uniformly on $\partial U$. This in turn immediately implies
Theorem~\ref{thm:location_res}.
\end{proof}

\section{Numerical Illustration}\label{sec:num}
In this section we will test the convergence of the rescaled spectrum towards the 
zeros of the polynomials $P_{n_1,n_2,n_3}$. The resonances are calculated
by finding the zeros of the Selberg zeta function with the same algorithm as
used by Borthwick \cite{Bor14} (see also \cite{JP02,GLZ04}) which has been 
implemented in python, using Sage \cite{Sag14} and the scipy/numpy \cite{Sci} package. 

\begin{figure}
\centering
        \includegraphics[width=1.1\textwidth]{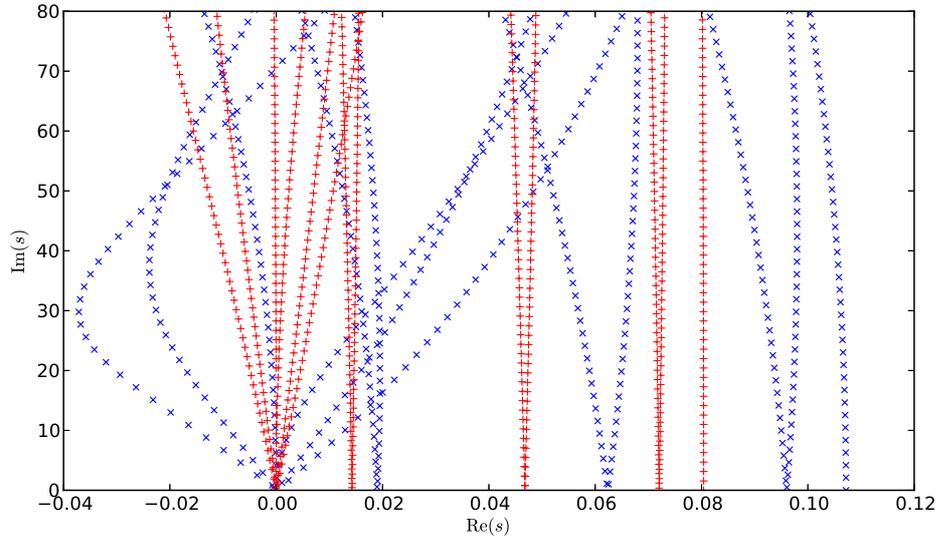}
\caption{Resonance spectrum for the surfaces $X_{4,4,5}(\ell)$ 
for $\ell=3$ (blue crosses) and $\ell=4$ (red plus signs).
}
\label{fig:res_445}
\end{figure}
In Figure~\ref{fig:res_445} we see the resonance spectrum of the surface
$X_{4,4,5}(\ell)$ for two different $\ell$-values ($\ell=3$ as blue crosses,
and $\ell=4$ in red plus signs). Both plots show significant resonance chains. However,
without rescaling, these chains are clearly different.The chains 
for $\ell=4$ are denser and are positioned at significantly smaller 
real part and are much less curved than the chains for $\ell=3$. Their rough 
structure is however very similar. From higher to lower real parts, both 
surfaces first have a single chain, then three pairs of chains that diverge
from each other and finally six resonance chains that emerge from $s=0$. This
common structure can be completely understood by the zeros of the polynomial
\[
P_{4,4,5}(z)= -4 \, z^{13} + z^{10} + 4 \, z^{9} + 4 \, z^{8} - 2 \, z^{5} - 4 \,
z^{4} + 1.
\]
\begin{figure}
\centering
        \includegraphics[width=1.1\textwidth]{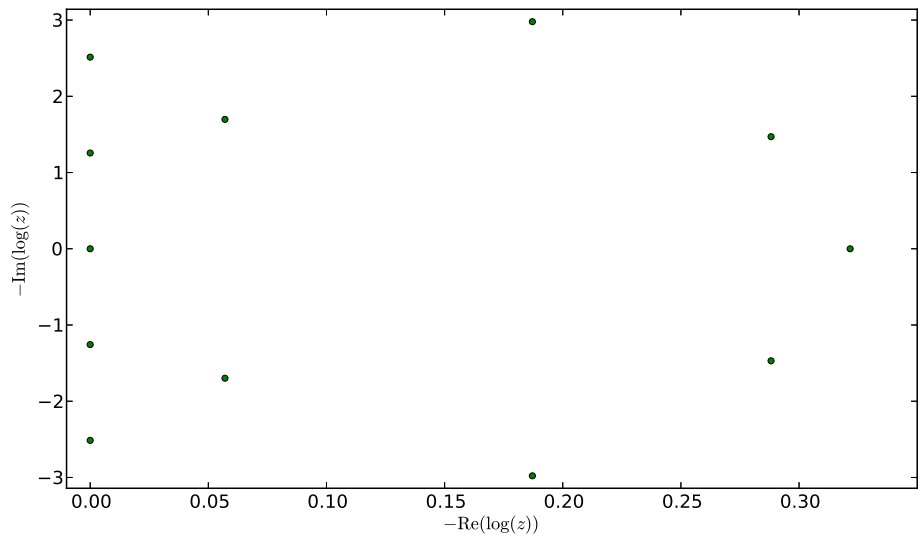}
\caption{Solutions of the equation $P_{4,4,5}(z)=0$ plotted on a
negative logarithmic scale. The zero at $\log(z)=0$ is of order two.
}
\label{fig:pol_zero_445}
\end{figure}

Figure~\ref{fig:pol_zero_445} shows the solutions of $P_{4,4,5}(z)=0$. 
As Theorem~\ref{thm:location_res} provides a connection between the 
resonances and the zeros of $P(e^{-s})$ we have plotted $-\log(z)$ in order
to compare the structure of the zeros directly with the resonances. And indeed 
the structure of the zeros of $P_{4,4,5}$ is exactly the same as the resonance chain 
structure. From higher to lower real parts (in the negative logarithmic plot 
of Figure~\ref{fig:pol_zero_445}) there is one leading zero, then three pairs
of zeros which have the same real part and finally 5 zeros with real part equal to
zero of which the zero with $\log(z)=0$ is of order two. But not only the 
rough resonance structure is described by the zeros
of $P_{4,4,5}$, also a large part of the rescaled spectrum is quantitatively well 
described by the zeros of $P_{4,4,5}(e^{-s})$. Figure~\ref{fig:rescaled_445} shows
the rescaled spectrum for $\ell=3$ and $\ell=4$. Additionally the zeros of 
$P_{4,4,5}(e^{-s})$ are plotted in green circles. One sees that in the plot
range already for $\ell=4$ the first 7 rescaled chains do very well 
coincide with their limit values given by $P_{4,4,5}$. Only the chains 
emerging from zero are still very unstable and show a visible difference. Overall 
however more then 70 resonances in the plot range are quantitatively well 
described by $P_{4,4,5}$. For $\ell=3$ the discrepancy is, as expected, higher
however all resonances on the first chain are also very well approximated. 
\begin{figure}
\centering
        \includegraphics[width=1.1\textwidth]{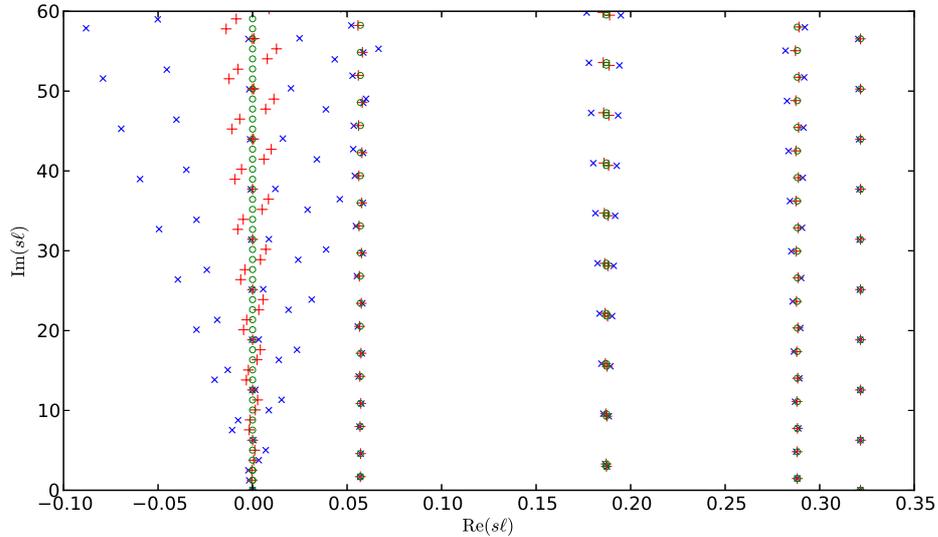}
\caption{Rescaled resonances for the surface 
$X_{4,4,5}(\ell)$ for $\ell=3$ (blue crosses) and $\ell=4$ (red plus signs). 
Additionally the green circles 
indicate the zeros of the polynomial $P_{4,4,5}(e^{-s})$ in the plot range. 
One observes that the resonances of the different surfaces 
really lie on approximately the same points after rescaling and that the 
zeros of the polynomial $P_{4,4,5}(e^{-s})$ predict the position of most of
the resonances for $\ell=4$ already very well.
}
\label{fig:rescaled_445}
\end{figure}

As a second example we show the same plots for the surfaces
$X_{4,5,6}(\ell)$, this time for $\ell=4$ and $\ell=5$ 
(see Figure~\ref{fig:res_456}, \ref{fig:pol_zero_456} and \ref{fig:rescaled_456}). 
The corresponding polynomial is now given by 
\[
 P_{4,5,6}(z)=-4 \, z^{15} + z^{12} + 2 \, z^{11} + 3 \, z^{10} + 2 \, z^{9} + z^{8} -
2 \, z^{6} - 2 \, z^{5} - 2 \, z^{4} + 1.
\]
\begin{figure}
\centering
        \includegraphics[width=1.1\textwidth]{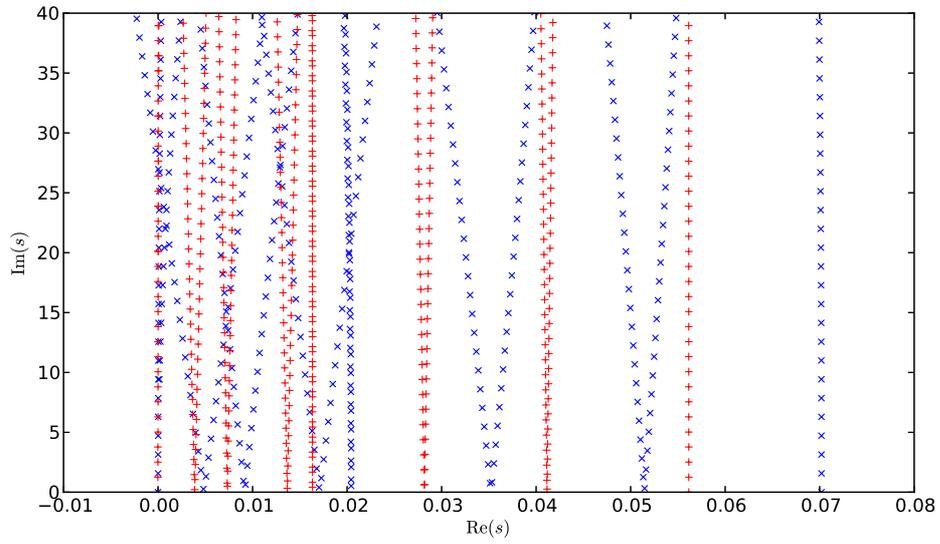}
\caption{Resonance spectrum for the surfaces $X_{4,5,6}(\ell)$ 
for $\ell=4$ (blue crosses) and $\ell=5$ (red plus signs).
}
\label{fig:res_456}
\end{figure}
\begin{figure}
\centering
        \includegraphics[width=1.1\textwidth]{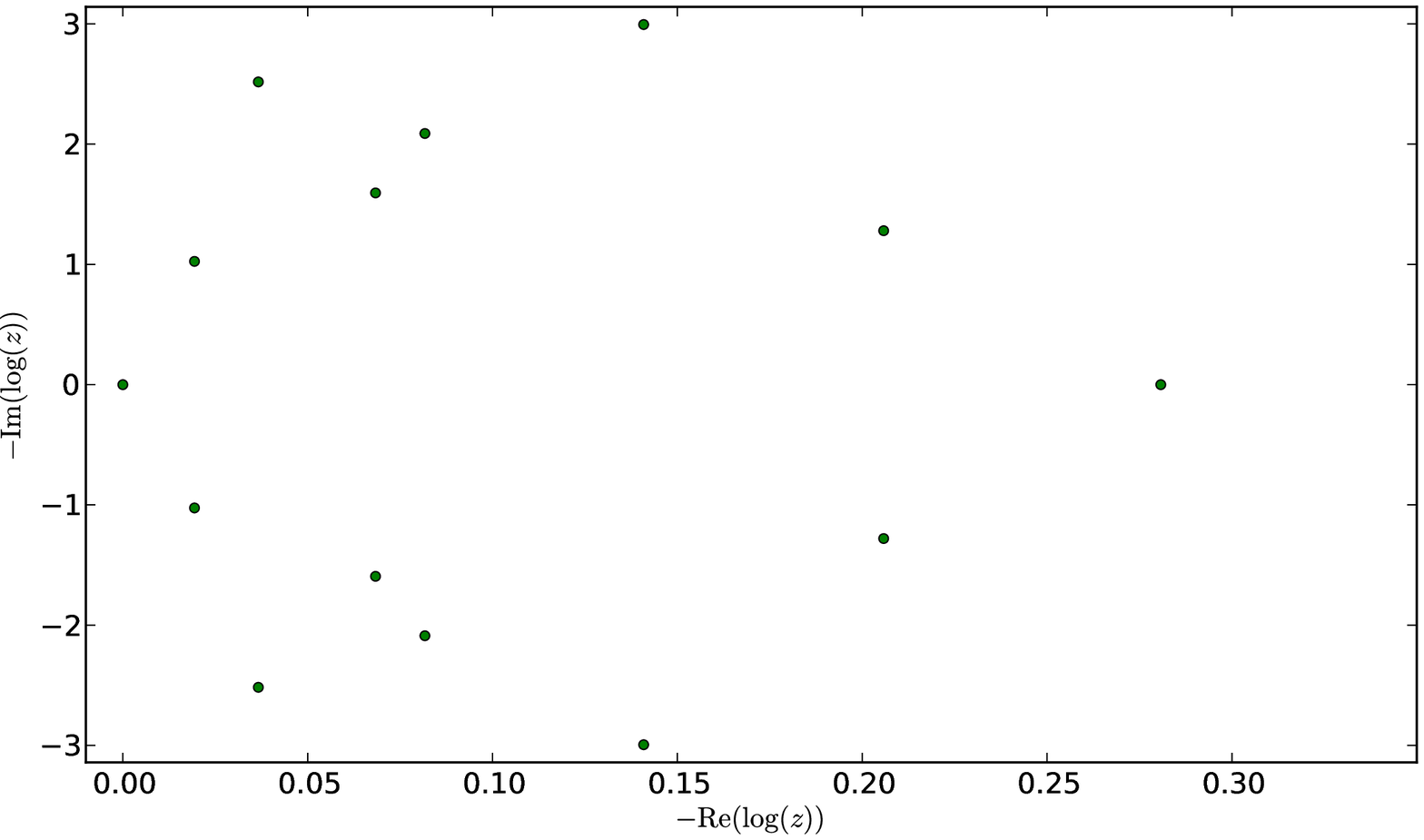}
\caption{Solutions of the equation $P_{4,5,6}(z)=0$ plotted on a
negative logarithmic scale. The zero at $\log(z)=0$ is of order two.
}
\label{fig:pol_zero_456}
\end{figure}

\begin{figure}
\centering
        \includegraphics[width=1.1\textwidth]{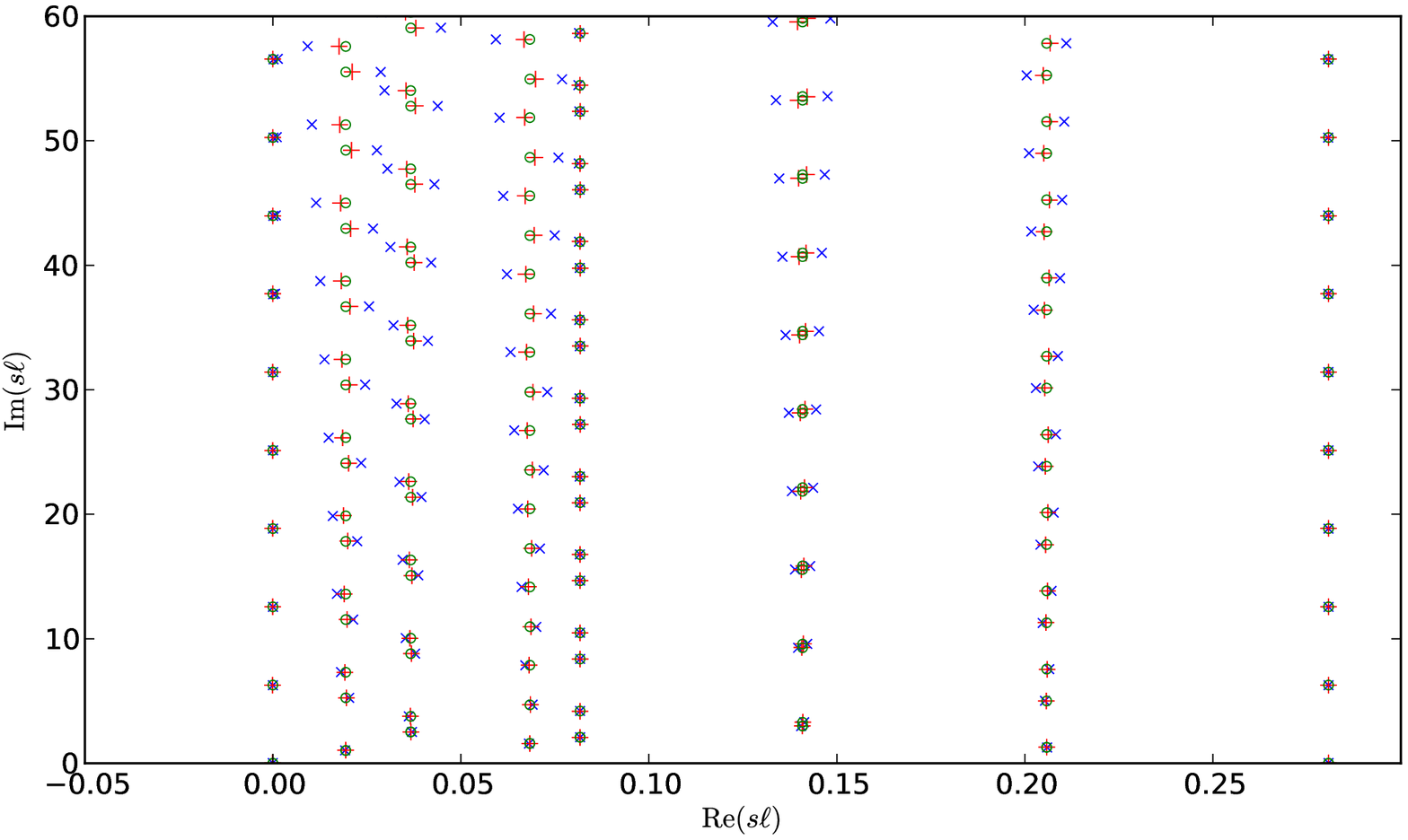}
\caption{Rescaled resonances for the surface 
$X_{4,5,6}(\ell)$ for $\ell=4$ (blue crosses) and $\ell=5$ (red plus signs). 
Additionally the green circles 
indicate the zeros of the polynomial $P_{4,5,6}(e^{-s})$ in the plot range. 
One observes that the resonances of the different surfaces 
lie on approximately the same points after rescaling and that the 
zeros of the polynomial $P_{4,5,6}(e^{-s})$ predict the position of most of
the resonances for $\ell=5$ already very well.
}
\label{fig:rescaled_456}
\end{figure}
As this surface is even less symmetric, the zeros of the polynomial has 
an even more complex structure (Figure~\ref{fig:pol_zero_456}). 
Now there is one leading zero, then 6 pairs
of zeros and finally a zero of order two at $\log(z)=0$. This corresponds 
exactly to the more complex chain structure with one leading chain and 7
further pairs of chains (Figure~\ref{fig:res_456}). Finally, after rescaling,
the position of a large part part of the plotted resonances agrees with 
the zeros of $P_{4,5,6}(e^{-s})$ (see Figure~\ref{fig:rescaled_456}).

Increasing the parameter $\ell$ even further yields a better and better 
coincidence between the numerically calculated resonances and those
predicted by the polynomial. For the surface $X_{12,12,12}$ we checked for 
example that the position of more then 150 individual resonances can be determined 
at a precision of $10^{-3}$ by calculating the zeros of the polynomial 
\[
 P_{1,1,1}(z)=-4 \, z^{3} + 9 \, z^{2} - 6 \, z + 1 = -(z-1)^2(4\,z-1)
\]
which can in this case even be factorized by hand.

\end{document}